\documentclass[12pt]{article}
\usepackage{amsmath}
\usepackage{graphicx} 
\usepackage{enumerate}
\usepackage{natbib}
\usepackage{url} 

\newcommand{\blind}{1}

\usepackage{amsfonts}
\usepackage{mathrsfs}
\usepackage{subfigure}
\usepackage{enumerate}
\usepackage{color}
\usepackage{amssymb}
\usepackage{amsthm}
\usepackage{bm}
\usepackage{here}
\usepackage{multicol}
\usepackage{xcolor}
\usepackage{hyperref} 

\hypersetup{linkbordercolor=green}

\newtheorem{thm}{Theorem}[section]
\newtheorem{lem}[thm]{Lemma}
\newtheorem{prop}[thm]{Proposition}
\theoremstyle{definition}
\newtheorem{dfn}[thm]{Definition}
\newtheorem{asp}{Assumption}[section]
\newtheorem{cor}[thm]{Corollary}
\newtheorem{rem}[thm]{Remark}
\newtheorem{eg}{Example}
\allowdisplaybreaks

\newcommand{\field}[1]{\mathbb{#1}} 
\newcommand{\C}{\mathbb{C}}

\newcommand{\I}{\,\mathrm{i}}
\newcommand{\N}{\field{N}}
\newcommand{\T}{\,\mathrm{T}}

\newcommand{\ve}{\,\mathrm{vec}}
\newcommand{\var}{\,\mathrm{var}}
\newcommand{\cov}{\,\mathrm{Cov}}

\newcommand*\dif{\mathop{}\!\mathrm{d}}

\newcommand\mi{\mathrm{i}}
\newcommand\me{\mathrm{e}}
\newcommand{\R}{\field{R}}

\newcommand{\Z}{\field{Z}}

\newcommand{\abs}[1]{\lvert#1\rvert}

\newcommand{\norm}[1]{\lVert#1 \rVert}
\newcommand{\Bnorm}[1]{\Bigl\lVert#1\Bigr  \rVert}
\newcommand{\Babs}[1]{\Bigl \lvert#1\Bigr \rvert} 
\newcommand{\ep}{\epsilon} 

\newcommand{\sumn}[1][i]{\sum_{#1 = 1}^T}
\newcommand{\tsum}[2][i]{\sum_{#1 = -#2}^{#2}}
\providecommand{\abs}[1]{\lvert#1\rvert}
\providecommand{\Babs}[1]{\Bigl \lvert#1\Bigr \rvert} 

\newcommand{\uint}{\int^{1}_{0}}
\newcommand{\freqint}{\int^{\pi}_{-\pi}}

\newcommand{\tr}{{\rm Tr}}
\newcommand{\cum}{{\rm cum}}

\newcommand{\xt}{\bm{X}_{t, T}}
\newcommand{\yt}{\bm{Y}_{t, T}}

\newcommand{\fu}[1][u]{\bm{f}(#1, \lambda)}
\newcommand{\ffu}{\bm{\mathfrak{f}}(u, \lambda)}

\newcommand{\btheta}{\bm{\theta}}

\newcommand{\bzero}{\bm{0}}
\newcommand{\bI}{\bm{I}}

\newcommand{\bx}[1]{\bm{X}_{#1, T}}
\newcommand{\be}{\bm{\ep}}
\newcommand{\bp}{\bm{\phi}}

\newcommand{\pf}{\bm{f}_{\btheta}(\lambda)}
\newcommand{\pfu}{\bm{f}_{\btheta(u)}(\lambda)}

\newcommand{\spf}[1]{\bm{f}_{\btheta}^{#1}(\lambda)}

\newcommand{\pff}{\mathfrak{f}_{\btheta}(\lambda)}

\newcommand{\dlim}{\xrightarrow{d}}
\newcommand{\plim}{\rightarrow_{P}}

\begin{document}

\def\spacingset#1{\renewcommand{\baselinestretch}%
{#1}\small\normalsize} \spacingset{1}


\if1\blind
{
  \title{\bf Statistical Inference for Local Granger Causality}
  \date{August 4, 2021}
  \author{Yan Liu\thanks{
    The author gratefully acknowledge \textit{JSPS Grant-in-Aid for Young Scientists (B) 17K12652
and JSPS Grant-in-Aid for Scientific Research (C) 20K11719}}\hspace{.2cm}\\
    Institute for Mathematical Science, Waseda University\\
    Masanobu Taniguchi\thanks{
    The author gratefully acknowledge \textit{JSPS Grant-in-Aid for Scientific Research (S) 18H05290}}\hspace{.2cm}\\
    Institute for Science and Engineering, Waseda University\\
    and \\
    Hernando Ombao\thanks{
    The author gratefully acknowledge \textit{the KAUST Research Fund}}\hspace{.2cm}\\
    Statistics Program, \\
    King Abdullah University of Science and Technology
    }
  \maketitle
} \fi

\if0\blind
{
  \bigskip
  \bigskip
  \bigskip
  \begin{center}
    {\LARGE\bf Title}
\end{center}
  \medskip
} \fi

\bigskip
\begin{abstract}
Granger causality has been employed to investigate 
causality relations between components of stationary multiple time series.
We generalize this concept by developing statistical inference for 
local Granger causality for multivariate locally stationary processes.
Our proposed local Granger causality approach captures time-evolving 
causality relationships in nonstationary processes.
The proposed local Granger causality is well represented in the frequency 
domain and estimated based on the parametric time-varying spectral density matrix using the 
local Whittle likelihood.
Under regularity conditions, we demonstrate that 
the estimators converge to multivariate normal in distribution.
Additionally,
the test statistic for the local Granger causality is shown to be asymptotically distributed 
as a quadratic form of a multivariate normal distribution.
The finite sample performance is confirmed with several simulation studies for multivariate 
time-varying autoregressive models.
For practical demonstration, the proposed local Granger causality method uncovered new 
functional connectivity relationships between channels in brain signals. Moreover, the 
method was able to identify structural changes in financial data. 
\end{abstract}

\noindent%

\spacingset{1.5} 
\section{Introduction}
Statistical inference for cause and effect remains at the forefront of many 
studies including biology, medicine, physical systems, environmental science, 
public health, policy and finance. However, there remain challenges on 
inference because causality is notoriously difficult to establish.
Granger causality, proposed in \cite{granger1963} and \cite{granger1969},
is a milestone of causal inference in dynamic models.
In broad terms, Granger causality from a time series $\{Y_t\}$ to 
another series $\{X_t\}$ measures the predictive ability from the series $\{Y_t\}$ to $\{X_t\}$.
If the predictive ability of $(X_s, Y_s)_{s<t}$ on ${X_t}$ is not different from the predictive ability of $(X_s)_{s<t}$ on ${X_t}$, then there is
``no Granger causal relationship'' 
from the series $\{Y_t\}$ to $\{X_t\}$.
Thus, Granger causality analysis is important for determining whether or not a set of variables
contains useful information for improving the prediction of another set of variables.

\cite{geweke1982} and \cite{geweke1984} considered measures of linear dependence and feedback between components of a multivariate time series
in both time and frequency domains.
\cite{hosoya1991} proposed a refinement of the above measures, which also has a well-defined representation in the frequency domain.
For nonstationary vector autoregressive (VAR) models, \cite{ssw1990} considered 
the Wald's statistic for the hypothesis testing problem and elucidated its nonstandard asymptotic distribution. 
Granger causality for nonstationary bivariate cointegrated processes has also been considered in \cite{gl1995}, among others.
A thorough treatment of the Granger causality for multivariate time series is discussed in \cite{l2005}.

In this paper, we propose a local Granger causality measure based on the locally stationary process.
A locally stationary process has a Cram{\'e}r-like representation but its  
transfer function is allowed to change over time. 
Formal models for the time-varying spectra of nonstationary processes 
have been developed since this concept was introduced in \cite{priestley1965}.
A more theoretically rigorous framework for multivariate locally stationary processes have been formulated in \cite{d2000}.
To estimate the time-varying spectral density of a locally stationary process,
\cite{ns1997} developed a wavelet estimator based on the pre-periodogram.
The parameter estimation for an evolutionary spectral is discussed in \cite{dg1998}.
Local inference for locally stationary time series was investigated in \cite{dahlhaus2009}.
Moreover, \cite{dp2009} constructed the estimation theory for the weak convergence of the empirical spectral processes.
To the best of our knowledge, despite the recent progress on models that capture nonstationary behavior, 
local Granger causality has not yet been developed.
To address this limitation, this paper undertakes the task of developing this local 
concept 
because many time series phenomena display Granger causality behavior 
that changes over the course of time
(e.g., electroencephalograms and stock market indices).
Thus, the contribution of this paper is a rigorous framework for statistical inference
for local Granger causality.

We focus on multivariate locally stationary processes
to develop the statistical inference for local Granger causality.
Statistical inference for multivariate stationary processes 
has been discussed in \cite{hannan1970}, \cite{tk2000}, \cite{ss2000} and references therein.
\cite{tpk1996} developed a nonparametric method to test the cross-relationships between multiple time series.
\cite{st2004} discussed the discriminant analysis for multivariate locally stationary processes
based on the likelihood ratio.
\cite{hos2004} proposed a SLEX model to develop a discriminant scheme that can extract local features of time series.
\cite{orsm2001} and \cite{osg2005}
developed models for bivariate and multivariate nonstationary data using the SLEX basis which consists 
of well-localized Fourier-like waveforms.

As noted, the goal of this paper is to develop statistical inference for local Granger causality 
for multivariate locally stationary processes.
In particular, local Granger causality is expressed in the frequency domain
using the foundational ideas on Granger causality
for stationary processes. 
We develop a procedure for parameter estimation based on the local Whittle likelihood
and 
derive the asymptotic distribution of the estimators.
Under regularity conditions, 
the estimates are shown to converge to multivariate normal in distribution.
Parametric Granger causality, however, converges to normal
or 
a quadratic form of normal random variables,
which depends on the gradient of the causality measure.
Several simulation studies were conducted to evaluate the finite sample performance 
for multivariate time-varying autoregressive models.
To illustrate the potential impact of the proposed work, we analyzed the log-returns of the financial data and
multichannel electroencephalogram (EEG) data. Using the proposed method, the local Granger causality analyses 
produced new insightful results in the data analyses.

The remainder of the paper is organized as follows.
In Section \ref{sec:2}, we propose the local Granger causality.
The properties of the local Granger causality are detailed immediately behind the definition.
In Section \ref{sec:3}, we develop statistical inference for local Granger causality
based on the local Whittle estimation for multivariate locally stationary processes.
Numerical results 
on finite sample performance of the estimator and the test statistic for local causality are reported in Section \ref{sec:4}.
In Section \ref{sec:5}, we apply the proposed local Granger causality to EEG data and financial data.
The proofs for the theoretical results in Section \ref{sec:3} are relegated to Section \ref{sec:6} in supplement.

\subsection{Notations}
$O_{m \times M}$ denotes a $m \times M$ zero matrix;
$I_p$ denotes the $p \times p$ identity matrix; 
For any matrix $A$, let $\norm{A}_{\infty}: = \max_{1 \leq i \leq p} \sum_{j = 1}^p \abs{a_{ij}}$.
For a square matrix $A$, $\abs{A}$ denotes its determinant. 
$\dlim$ denotes the convergence in distribution.
Additionally, let $l$ be a function such that
\[
l(j): = 
\begin{cases}
1, &\qquad \abs{j} \leq 1, \\
\abs{j} \log^{1 + \kappa}\abs{j}, &\qquad \abs{j} > 1,
\end{cases} 
\]
for some constant $\kappa > 0$.

\section{Local Granger Causality} \label{sec:2}
In this section, we introduce the 
concept of local Granger causality (LGC) in the framework of locally stationary processes.
Let $\xt = (X_{t, T}^{(1)}, \dots, X_{t, T}^{(p)})^{\T}$ be a sequence of $p$-dimensional multivariate stochastic processes
\begin{equation} \label{eq:2.1}
\bx{t} = \tsum[j]{\infty} A_{t, T}(j) \be_{t - j},
\end{equation}
where the sequences $\{A_{t, T}(j)\}_{j \in \Z}$ satisfy the following conditions:
there exists a positive constant $C_A$ such that
\[
\sup_{t, T} \norm{A_{t, T}(j)}_{\infty} \leq \frac{C_A}{l(j)},
\]
and there exists a sequence of functions $A(\cdot, j): [0, 1] \to \R$ such that
\begin{enumerate}[(i)]
\item $\sup_{u} \norm{A(u, j)}_{\infty} \leq \frac{C_A}{l(j)}$;
\item $\sup_j \Bnorm{A_{t, T}(j) - A\Bigl(\frac{t}{T}, j\Bigr)}_{\infty} \leq \frac{C_A}{l(j)} T^{-1}$;
\item $V\Bigl(\norm{A(\cdot, j)}_{\infty}\Bigr) \leq \frac{C_A}{l(j)}$,
\end{enumerate}
where $V(f)$ is the total variation of the function $f$ on the interval $[0, 1]$, i.e., $V$  is defined as
\[
V(f) = \sup\Bigl\{
\sum_{k = 1}^m \abs{f(x_k) - f(x_{k - 1})}; \,
0 \leq x_0 < \cdots < x_m \leq 1, \, m \in \N
\Bigr\}.
\]
The process \eqref{eq:2.1} is usually referred to as the multivariate locally stationary process.
We impose the following assumptions on the process \eqref{eq:2.1} for the estimation theory later on.

\begin{asp} \label{asp:2.1}
For the process in \eqref{eq:2.1},
let  $\be_t$ be independent and identically distributed with
$E \be_t = \bzero$ and $E \be_t \be_t^{\top} = K$,
where the matrix $K$ exists and all elements are bounded by $C_K$.
Furthermore, all elements in the $r$th moment of $\be_t$ exist and bounded by $C_{\bm{\ep}}^{(r)}$.
There exists a finite constant $C > 0$ such that  $C_{\bm{\ep}}^{(r)} < C$.
\end{asp}


Let $m$ and $M$ be two positive integers such that $p = m + M$.
Suppose $\bm{X}_{t, T} = \bigl({\bm{X}_{t, T}^{(1)}}^{\top}, {\bm{X}_{t, T}^{(2)}}^{\top}\bigr)^{\top}$,
$\bm{X}_{t, T}^{(1)} \in \R^m$, $\bm{X}_{t, T}^{(2)} \in \R^M$, 
has the time-varying spectral density matrix $\fu$ with the partition
\begin{equation} \label{eq:2.11}
\fu = 
\begin{pmatrix}
\fu_{11} & \fu_{12} \\
\fu_{21}& \fu_{22}
\end{pmatrix}
:=
\frac{1}{2\pi} A(u, \lambda) K A(u, -\lambda)^{\top}, \qquad u \in [0, 1],
\end{equation}
where $A(u, \lambda): = \sum_{j = -\infty}^{\infty} A(u, j) \exp(\I j \lambda)$.
Let  $\Sigma(u)$ be the one-step-ahead prediction error covariance matrix
based on the time-varying spectral density matrix $\fu$ with the same partition.
By the Kolmogorov's formula for multiple time series, we have
\[
\det \Sigma(u) = \exp\Biggl(
\frac{1}{2\pi} \freqint
\log \Bigl(
\det 2 \pi \fu
\Bigr)
\dif \lambda
\Biggr), \qquad \text{for any $u \in [0, 1]$},
\]
(see \cite{hannan1970}, p.162).

Let $H(\tau_1, \tau_2) = \overline{\text{sp}}(\xt^{(1)}, 1 \leq t \leq \tau_1; \xt^{(2)}, 1 \leq t \leq \tau_2)$
be the closed linear subspace generated by 
$\{
\xt^{(1)}, 1 \leq t \leq \tau_1;
\xt^{(2)}, 1 \leq t \leq \tau_2
\}$.
Especially, we use $H(\tau_1, 0)$ and $H(0, \tau_2)$ to express the 
closed linear subspace generated by 
$\{
\xt^{(1)}, t \leq \tau_1
\}$,
$\{
\xt^{(2)}, t \leq \tau_2
\}$,
respectively.

Introducing a companion process
\begin{equation} \label{eq:r2.6}
\yt^{(2)} = 
\xt^{(2)} - E \bigl(\xt^{(2)} \mid H(t, t-1)\bigr),
\end{equation}
we propose the local Granger causality measure from $\{\xt^{(2)}\}$ to $\{\xt^{(1)}\}$ as
\begin{equation} \label{eq:2.15}
{\rm GC}^{(2 \to 1)}(u) = 
\frac{1}{2\pi} \freqint
{\rm FGC}(u, \lambda) d\lambda,
\end{equation}
where
\begin{equation*} 
{\rm FGC}(u, \lambda)  = 
\log \frac{ \abs{\fu_{11}}}
{\Babs{\fu_{11} - 2\pi \ffu_{12} \tilde{\Sigma}(u)_{22}^{-1} \ffu_{21} }}.
\end{equation*}
Here, $\ffu$ is the time-varying spectral density matrix of the process $\{\bigl({\bm{X}_{t, T}^{(1)}}^{\top}, {\bm{Y}_{t, T}^{(2)}}^{\top}\bigr)^{\top}\}$
(see Propositions \ref{prop:2.1} and \ref{prop:2.2} below),
and $\tilde{\Sigma}(u)$ is an $(M \times M)$-matrix 
\begin{equation*} 
\tilde{\Sigma}(u)_{22}
=
\Sigma(u)_{22} - \Sigma(u)_{21} \Sigma(u)_{11}^{-1} \Sigma(u)_{12}.
\end{equation*}

The proposal of the local Granger causality \eqref{eq:2.15} is motivated by Hosoya's measure of causality 
(a term first coined in \cite{gl1995}) 
in combination with the nonstationary version of Kolmogorov's formula by \cite{dahlhaus1996}.
The companion process $\{\bm{Y}_{t, T}^{(2)}\}$ in \eqref{eq:r2.6} is introduced in order to remove the possible effect brought by 
the nonorthogonality between residuals of predictions
$E\bigl(\xt^{(1)} \mid H(t - 1, t - 1) \bigr)$ and $E\bigl(\xt^{(2)} \mid H(t - 1, t - 1) \bigr)$.

We now start to explain properties of the proposed local Granger causality.
Denote by $\Sigma_{t, T}$ the one-step-ahead prediction error covariance matrix, i.e.,
\[
\Sigma_{t, T} = \var\bigl[\bm{X}_{t, T} - E\bigl(\xt \mid H(t - 1, t - 1) \bigr)\bigr].
\]

\begin{prop} \label{prop:2.1}
The companion process $\{\bm{Y}_{t, T}^{(2)}\}$ is a locally stationary process with the time-varying spectral density
\begin{equation} \label{eq:r3.2.6}
\ffu_{22} = \frac{1}{2\pi} \tilde{\Sigma}(u)_{22}.
\end{equation}
\end{prop}

\begin{proof}
From the definition of $\{\bx{t}\}$ in \eqref{eq:2.1}, we have
\[
\xt - E \bigl(\xt \mid H(t - 1, t - 1)\bigr) = A_{t, T}(0) \bm{\ep}_t.
\]
By the following formula (see Lemma 2.2 in \cite{hosoya1991})
\begin{multline*}
\bm{X}_{t, T}^{(2)} - E\bigl(\bm{X}_{t, T}^{(2)} \mid H(t, t - 1) \bigr)
= \\
\bigl\{\bm{X}_{t, T}^{(2)} - E\bigl(\bm{X}_{t, T}^{(2)} \mid H(t - 1, t - 1) \bigr)\bigr\}
-
\Sigma_{t, T, 21} \Sigma_{t, T, 11}^{-1}
\bigl\{\bm{X}_{t, T}^{(1)} - E\bigl(\bm{X}_{t, T}^{(1)} \mid H(t - 1, t - 1) \bigr)\bigr\},
\end{multline*}
we find that
\begin{equation*}
\yt^{(2)} = \bm{\ep}_t^{(2)} - \Sigma_{t, T, 21} \Sigma_{t, T, 11}^{-1} \bm{\ep}_t^{(1)}\\
=
\begin{pmatrix}
- \Sigma_{t, T, 21} \Sigma_{t, T, 11}^{-1} & \,\,\, \bm{I}_M
\end{pmatrix}
\bm{\ep}_t.
\end{equation*}
In view of Example 2.3 (i) in \cite{d2000}, $\{\bm{Y}_{t, T}^{(2)}\}$ is locally stationary.
A straightforward calculation gives the expression of $\ffu_{22}$ in \eqref{eq:r3.2.6}. 
\end{proof}

Let $\mathcal{H}^{(2)}(\tau)$ be the closed linear subspace generated by $\{\yt^{(2)}, 1 \leq t \leq \tau\}$.
The Hosoya measure is defined as
\begin{equation} \label{eq:2.8}
{\rm HM}^{(2 \to 1)}_{t, T}: =
\log
\frac{
\det \var\bigl[ \xt^{(1)} - 
E\bigl(\xt^{(1)} \mid H(t - 1, 0) \bigr) \bigr]}
{
\det \var\bigl[ \xt^{(1)} - 
E\bigl(\xt^{(1)} \mid \sigma\{H(t - 1, 0) \cup \mathcal{H}^{(2)}(t - 1)\}
\bigr)
\bigr]}.
\end{equation}
For any fixed $u \in [0, 1]$, the time-varying spectral matrix $\fu$ has a factorization 
\begin{equation} \label{eq:r3.3.6}
\bm{f}(u, \lambda) 
= 
\frac{1}{2 \pi} \bm{\Lambda}(u, \me^{-\mi\lambda}) \bm{\Lambda}(u, \me^{\mi\lambda})^*,
\quad
z \in \mathcal{D},
\end{equation}
(see \cite{rozanov1967}).

\begin{prop} \label{prop:2.2}
Suppose all eigenvalues of $A(u, \lambda) A(u, -\lambda)^{\top}$ are bounded from below by some constant $C > 0$
uniformly in $u$ and $\lambda$,
and all components of $A(u, \lambda)$ are differentiable in $u$ and $\lambda$ with bounded derivatives 
$(\partial/\partial u)(\partial/\partial \lambda) A(u, \lambda)_{ab}$ for $a, b \in \{1, 2, \dots, p\}$.
It holds that
\[
\abs{{\rm GC}^{(2 \to 1)}(t/T) - {\rm HM}^{(2 \to 1)}_{t, T}} = o_t(1) + O_T(1),
\]
where the $o_t(1)$ term is uniform in $T$ and the $o_T(1)$ term is uniform in $t$.
\end{prop}

\begin{proof}
From Proposition \ref{prop:2.1}, we see that 
the process $\{\bigl({\bm{X}_{t, T}^{(1)}}^{\top}, {\bm{Y}_{t, T}^{(2)}}^{\top}\bigr)^{\top}\}$ is locally stationary.
In view of Lemma 2.3 in \cite{hosoya1991}, we see that the process has the time-varying spectral density matrix $\ffu$ with
$\ffu_{11} = \fu_{11}$ and 
\begin{equation} \label{eq:r4.2.9}
\ffu_{21} = \bm{\frak{f}}(u, -\lambda)_{12}^{\top} = \Bigl(
-\Sigma(u)_{21} \Sigma(u)_{11}^{-1} \quad I_M
\Bigr) \bm{\Lambda}(u, 0) \bm{\Lambda}(u, \me^{\mi \lambda})^{-1}
\begin{pmatrix}
\fu_{11}  \\
\fu_{12}
\end{pmatrix}.
\end{equation}
A direct computation shows that the process 
$\bigl\{\xt^{(1)} - 
E\bigl(\xt^{(1)} \mid \sigma\{H(t - 1, 0) \cup \mathcal{H}^{(2)}(t - 1)\}
\bigr)\bigr\}$
is still locally stationary and has the time-varying spectral density 
\[
\ffu_{11} - \ffu_{12} \ffu_{22}^{-1} \ffu_{21}.
\]
Inspection of Theorem 3.2 in \cite{dahlhaus1996} for the nonstationary version of Kolmogorov's formula reveals that
\begin{multline} \label{eq:r3.2.9}
\det \var\bigl[ \xt^{(1)} - 
E\bigl(\xt^{(1)} \mid H(t - 1, 0) \bigr) \bigr]
=\\
\exp\Biggl(
\frac{1}{2\pi} \freqint
\log \Bigl(
\det 2 \pi \bm{f}(t/T, \lambda)_{11}
\Bigr)
\dif \lambda
\Biggr) + o_t(1) + o_T(1),
\end{multline}
and
\begin{multline} \label{eq:r3.2.10}
\det \var\bigl[ \xt^{(1)} - 
E\bigl(\xt^{(1)} \mid \sigma\{H(t - 1, 0) \cup \mathcal{H}^{(2)}(t - 1)\}
\bigr)
\bigr]\\
=
\exp\Biggl(
\frac{1}{2\pi} \freqint
\log \Bigl(
\det 2 \pi \bigl(\bm{\frak{f}}(t/T, \lambda)_{11} - \bm{\frak{f}}(t/T, \lambda)_{12} 
\bm{\frak{f}}(t/T, \lambda)_{22}^{-1} \bm{\frak{f}}(t/T, \lambda)_{21}\bigr)
\Bigr)
\dif \lambda 
\Biggr) \\+ o_t(1) + o_T(1).
\end{multline}
Combining \eqref{eq:r3.2.9} and \eqref{eq:r3.2.10} yields the desired result. 
\end{proof}

The local Granger causality measure can be regarded as the limit of that constructed by the Wigner-Ville spectrum.
To be specific, let $\bm{f}_{t, T}(\lambda)$ be the Wigner-Ville spectrum of the process $\{\bm{X}_{t, T}\}$, i.e.,
\[
\bm{f}_{t, T}(\lambda) = 
\frac{1}{2\pi}
\sum_{s = - \infty}^{\infty}
\cov\bigl(
\bm{X}_{[t - s/2], T}, \bm{X}_{[t + s/2], T} 
\Bigr)\exp(-\mi \lambda s),
\]
(see \cite{mf1985}).
The measure of the Wigner-Ville spectrum now is
\[
{\rm GC}^{(2 \to 1)}_{t, T}
=
\frac{1}{2\pi} \freqint
{\rm FGC}_{t, T}(\lambda) \dif \lambda,
\]
where
\begin{equation*} 
{\rm FGC}_{t, T}(\lambda)  = 
\log \frac{ \abs{\bm{f}_{t, T}(\lambda)_{11}}}
{\Babs{\bm{f}_{t, T}(\lambda)_{11} - 2\pi \bm{\frak{f}}_{t, T}(\lambda)_{12} 
\tilde{\Sigma}_{t, T, 22}^{-1} 
\bm{\frak{f}}_{t, T}(\lambda)_{21} }}
\end{equation*}
and $\tilde{\Sigma}(u)$ is an $(M \times M)$-matrix 
\begin{equation*} 
\tilde{\Sigma}_{t, T, 22}
=
\Sigma_{t, T, 22} - \Sigma_{t, T, 21} \Sigma_{t, T, 11}^{-1} \Sigma_{t, T, 12}.
\end{equation*}

\begin{prop}
Suppose $\fu$ is uniformly Lipschiz continuous with respect to $u$ and $\lambda$.
For any sequence $t/T \to u$, 
We have 
\[
\Babs{
{\rm GC}^{(2 \to 1)}(u) - {\rm GC}^{(2 \to 1)}_{t, T}
}
=
o(1).
\]
\end{prop}

\begin{proof}
We only show that 
\[
\frac{1}{2\pi} \freqint 
\bigl(
\log \abs{\fu_{11}} - \log \abs{\bm{f}_{t, T}(\lambda)_{11}}
\bigr)
\dif \lambda
=
o(1),
\]
to see the difference in the numerator. The denominator can be proved similarly.

Let us consider the scalar process $\bm{\alpha}^{*} \bm{X}_{t, T}^{(1)}$ 
with the Wigner-Ville spectrum 
$f^{\bm{\alpha}}_{t, T}(\lambda) :=\bm{\alpha}^{*} \bm{f}_{t, T}(\lambda)_{11} \bm{\alpha}$ for any $\bm{\alpha} \in \C^m$.
Comparing it with $f^{\bm{\alpha}}(u, \lambda) = \bm{\alpha}^{*} \bm{f}(u, \lambda)_{11} \bm{\alpha}$, 
by Theorem 2.2 in \cite{dahlhaus1996}, we see that
\[
\freqint \abs{f_{t, T}^{\bm{\alpha}}(\lambda) - f^{\bm{\alpha}}(u, \lambda)}^2 \dif \lambda = o(1),
\]
which implies that
\begin{equation} \label{eq:r3.2.11}
\freqint \abs{f_{t, T}^{\bm{\alpha}}(\lambda) - f^{\bm{\alpha}}(u, \lambda)} \dif \lambda = o(1),
\end{equation}
since by the Cauchy-Schwarz inequality, we have
\[
\freqint \abs{f_{t, T}^{\bm{\alpha}}(\lambda) - f^{\bm{\alpha}}(u, \lambda)} \dif \lambda
\leq 
\sqrt{2 \pi} \Bigl(
\freqint \abs{f_{t, T}^{\bm{\alpha}}(\lambda) - f^{\bm{\alpha}}(u, \lambda)}^2 \dif \lambda
\Bigr)^{1/2}.
\]
By Taylor's expansion, we have
\begin{multline} \label{eq:r3.2.12}
\log\abs{\bm{f}_{t, T}(\lambda)_{11}}
=
\log \abs{\fu_{11}}
+
\tr\Bigl[
\fu_{11}^{-1} \bigl(
\bm{f}_{t, T}(\lambda)_{11} - \fu_{11}
\bigr)
\Bigr] \\
+
o\Bigl(
\tr\Bigl[
\fu_{11}^{-1} \bigl(
\bm{f}_{t, T}(\lambda)_{11} - \fu_{11}
\bigr)
\Bigr]
\Bigr).
\end{multline}
Remembering that $\bm{f}(u, \lambda) 
= 
\frac{1}{2 \pi} \bm{\Lambda}(u, \me^{-\mi\lambda}) \bm{\Lambda}(u, \me^{\mi\lambda})^*$ from \eqref{eq:r3.3.6}, 
we see that there exists an $m \times m$ Hermitian matrix $\bm{B}$ such that $\bm{f}(u, \lambda)^{-1}_{11} = \bm{B}^* \bm{B}$, 
and thus
\begin{equation} \label{eq:r3.2.13}
\tr\Bigl[
\fu_{11}^{-1} \bigl(
\bm{f}_{t, T}(\lambda)_{11} - \fu_{11}
\bigr)
\Bigr] 
=
\tr\Bigl[ \bm{B}
\bigl(
\bm{f}_{t, T}(\lambda)_{11} - \fu_{11}
\bigr) \bm{B}^*
\Bigr],
\end{equation}
which is a sum of quadratic forms $f_{t, T}^{\bm{\alpha}}(\lambda) - f^{\bm{\alpha}}(u, \lambda)$.
Applying \eqref{eq:r3.2.11} to \eqref{eq:r3.2.13} yields 
\[
\freqint \tr\Bigl[
\fu_{11}^{-1} \bigl(
\bm{f}_{t, T}(\lambda)_{11} - \fu_{11}
\bigr)
\Bigr] 
\dif \lambda
=o(1),
\]
and by observing \eqref{eq:r3.2.12}, we conclude that
\[
\frac{1}{2\pi} \freqint 
\log \abs{\fu_{11}} - \log \abs{\bm{f}_{t, T}(\lambda)_{11}}
\dif \lambda
=
o(1). 
\] 
\end{proof}

\begin{rem}
The construction of a linear predictor in practice for the locally stationary process may be of interest to some readers.
It can be shown that the predictor for the locally stationary process and that for the stationary approximation
are asymptotically equivalent under adequate conditions.
In contrast, we focus on the nonstationary version of Kolmogorov's formula found in \cite{dahlhaus1996}.
We elucidated that our local causality measure, as the limit of the measure constructed by the Wigner-Ville spectrum,
is a unique measure for multivariate locally stationary processes.
\end{rem}

\section{Statistical Inference for Local Granger Causality} \label{sec:3}
In this section, we develop the foundations for  
statistical inference for local Granger causality. 
The  proofs of the theoretical results are relegated to Section \ref{sec:6} in supplement. 


\subsection{Local Whittle estimation}
Let $\{\xt\}$ be the multivariate locally stationary process defined by \eqref{eq:2.1}
with the time-varying spectral density $\fu$ defined by \eqref{eq:2.11}.
The starting point is local estimation
by fitting a parametric spectral density model $\pf$,
$\bm{\theta} \in \Theta \subset \R^d$, to $\fu$.

Consider the observation stretch $(\bm{X}_{1, T}, \dots, \bm{X}_{T, T})$ and define 
$\bI_T(u, \lambda)$ to be the pre-periodogram matrix
\begin{equation} \label{eq:r3.17}
\bI_T(u, \lambda)
=
\frac{1}{2\pi} \sum_{l: 1 \leq [u T + 1/2 \pm l/2] \leq T}
\bx{[u T + 1/2 + l/2]}
\bx{[u T + 1/2 - l/2]}^{\top}
\exp(- \I \lambda l).
\end{equation}
Note that the pre-periodogram $\bI_T(u, \lambda)$ was first introduced in \cite{ns1997}.

We define the spectral divergence $\mathcal{L}(\btheta, u)$ between the parametric spectral density
and the time-varying spectral density as
\begin{equation} \label{eq:r15}
\mathcal{L}(\btheta, u)
=
\freqint 
\log \det \pf +
\tr \Bigl(
\fu
\spf{-1}
\Bigr)
\dif \lambda.
\end{equation}
For any fixed $u \in [0, 1]$, define $\bm{\theta}_0(u)$ as
\begin{equation} \label{eq:3.1}
\bm{\theta}_0(u): = \arg\min_{\btheta \in \Theta} \mathcal{L}(\btheta, u).
\end{equation}
Let $u_k: = k/T$. The sample analogue $\mathcal{L}_T$ of the spectral divergence is defined as 
\begin{equation} \label{eq:r17}
\mathcal{L}_T(\btheta, u)
=
\frac{1}{T} \sum_{k=1}^T
\frac{1}{b_T}
K\Bigl(
\frac{u - u_k}{b_T}
\Bigr) \\
\freqint 
\log \det \pf +
\tr \Bigl(\bI_{T}(u_k, \lambda) 
\spf{-1}
\Bigr)
\dif \lambda,
\end{equation}
and the local Whittle estimator of $\hat{\bm{\theta}}_T(u)$ is defined as
\begin{equation} \label{eq:3.1}
\hat{\btheta}_T(u): = \arg\min_{\btheta \in \Theta} \mathcal{L}_T(\btheta, u).
\end{equation}

We impose the following assumptions on the class of time-varying spectral densities
and the kernel function $K$ in \eqref{eq:r17}
to investigate the asymptotic properties of the local Whittle estimator \eqref{eq:3.1}.
\vspace{-0.2cm}
\begin{asp} \label{asp:3.1} \mbox{}
\begin{enumerate}[{\rm (i)}]
\item
The time-varying spectral density matrix $\fu$ is continuously differentiable 
with respect to $u$ for $u \in (0, 1)$.
\item
$K: \R \to \R$ is a nonnegative, bounded symmetric continuous function of bounded variation 
with a compact support $[-1, 1]$ satisfying $\int K(x) \dif x = 1$.
Let 
\[
K_b(x) := \frac{1}{b} K \Bigl(\frac{x}{b}\Bigr),
\]
where $b := b_T \to 0$, as $T \to \infty$.
\end{enumerate}
\end{asp}
We now specify the regularity conditions for the parametric model $\pf$
and the local parameter $\btheta(u)$.
For the brevity, let $\btheta := \btheta(u)$ when $u$ does not matter.
\begin{asp} \label{asp:3.2} \mbox{}
\begin{enumerate}[{\rm (i)}]
\item
For any fixed $u \in [0, 1]$, $\btheta(u) \in \Theta$,
where $\Theta$ is a compact subset of $\R^d$.
\item
For any fixed $u \in [0, 1]$, 
$\bm{f}_{\btheta^{(1)}(u)} \not = \bm{f}_{\btheta^{(2)}(u)}$
on a set of positive Lebesgue measure, if $\btheta^{(1)}(u) \not = \btheta^{(2)}(u)$.
\item
The parametric spectral density matrix $\pf$ is bounded away from 0 for each component,
and is continuously differentiable with respect to $\lambda$ for 
$\lambda \in (-\pi, \pi)$.
\item
For any $\btheta \in \Theta$,
$\bm{f}_{\btheta}$ is positive definite and
it is twice continuously differentiable with respect to $\btheta$.
\item For any fixed $u \in [0, 1]$, 
\begin{enumerate}[{\rm ({v-}a)}]
\item $\btheta_0(u) \in \Theta$ is the unique minimizer of $\mathcal{L}(\btheta, u)$ and 
lies in the interior of $\Theta$.
\item the matrices
\[
M_I^u
=
\freqint
\Biggl[
\frac{\partial^2}{\partial \btheta \partial \btheta^{\top}}
\tr\Bigl\{
\spf{-1} \bI_T(u, \lambda)
\Bigr\}
+
\frac{\partial^2}{\partial \btheta \partial \btheta^{\top}}
\log \det
\pf
\Biggr] \dif \lambda
\] 
and
\[
M_f^u
=
\freqint
\Biggl[
\frac{\partial^2}{\partial \btheta \partial \btheta^{\top}}
\tr\Bigl\{
\spf{-1} \fu
\Bigr\}
+
\frac{\partial^2}{\partial \btheta \partial \btheta^{\top}}
\log \det
\pf
\Biggr] \dif \lambda
\] 
are both positive definite.
\end{enumerate}
\end{enumerate} 
\end{asp}

First, let us consider the asymptotics for the sample analog of the spectral divergence $\mathcal{L}_T(\btheta, u)$.
\begin{thm} \label{thm:3.2}
Suppose Assumptions \ref{asp:2.1}, \ref{asp:3.1} and \ref{asp:3.2} hold.
For any $u \in (0, 1)$, if $b_T^{-1} = o\bigl(T(\log T)^{-6}\bigr)$ and $b_T = o(T^{-1/5})$, then we have
\begin{equation*}
\sqrt{T b_T}
\bigl(
\mathcal{L}_T(\btheta, u) - 
\mathcal{L}(\btheta, u)
\bigr)
\dlim
\mathcal{N}(0, \mathbb{V}^{\mathcal{L}}(u)).
\end{equation*}
as $T \to \infty$,
where
\begin{multline} \label{eq:3.2}
\mathbb{V}^{\mathcal{L}}(u) =
4\pi \int_{-1}^1 K(v)^2 \dif v 
\Biggl(
\freqint 
\tr \Bigl(
\fu[u] \spf{-1} \fu[u] \spf{-1}
\Bigr)\dif \lambda\\
+
\frac{1}{2}
\sum_{r, t, v ,w = 1}^p
\freqint \freqint 
\Bigl(
\bm{f}_{\btheta}^{rt}(\lambda_1)
\bm{f}_{\btheta}^{vw}(\lambda_2)
\tilde{\gamma}_{rtvw} (u; -\lambda_1, \lambda_2, -\lambda_2)
\Bigr) \dif \lambda_1 \dif \lambda_2
\Biggr),
\end{multline}
where $\tilde{\gamma}$ is the fourth-order spectral density of the process.
\end{thm}
\begin{rem}
Inspection of the proof of Theorem \ref{thm:3.2} reveals  that the main order of bias is $O(T^{1/2} b_T^{5/2})$
and the asymptotic variance is of order $O(T^{-1} b_T^{-1})$.
The optimal order of the bandwidth $b_T$ can be determined by equating squared bias and variance.
Thus, we obtain the optimal order $b_T = O(T^{-1/3})$ and the mean square error is $O(T^{-2/3})$.
This optimal order is similar to the one derived in \cite{kunsch1989} in the context of 
statistical inference for stationary time series.
\end{rem}
Let $M_{f,0}^u$ be
\begin{equation*}
M_{f,0}^u
=
\freqint
\Biggl[
\frac{\partial^2}{\partial \btheta \partial \btheta^{\top}}
\tr\Bigl\{
\spf{-1} \fu
\Bigr\}
+
\frac{\partial^2}{\partial \btheta \partial \btheta^{\top}}
\log \det
\pf
\Biggr]_{\btheta= \btheta_0(u)} \dif \lambda.
\end{equation*}
Now we establish the asymptotic normality of the local estimator $\hat{\btheta}_T(u)$.

\begin{thm} \label{thm:3.3}
Suppose Assumptions \ref{asp:2.1}, \ref{asp:3.1} and \ref{asp:3.2} hold.
For any $u \in (0, 1)$, if $b_T^{-1} = o\bigl(T(\log T)^{-6}\bigr)$ and $b_T = o(T^{-1/5})$, then we have
\begin{equation} \label{eq:3.3}
\sqrt{T b_T}
\bigl(
\hat{\btheta}_T(u) - 
\btheta_0(u)
\bigr)
\dlim
\mathcal{N}(\bm{0}, \mathbb{V}(u)),
\end{equation}
as $T \to \infty$, where $\mathbb{V}(u):= (M_{f,0}^u)^{-1} \mathbb{V}^{\btheta}(u)  (M_{f,0}^u)^{-1}$ and
\begin{multline} \label{eq:3.4}
\mathbb{V}^{\btheta}(u)_{ab}
=
4\pi \int_{-1}^1 K(v)^2 \dif v 
\Biggl(
\freqint 
\tr \Bigl[
\fu[u] \biggl\{\frac{\partial}{\partial \theta_a}\spf{-1}\biggr\} \fu[u] 
\biggl\{\frac{\partial}{\partial \theta_b} \spf{-1}\biggr\}
\Bigr]\dif \lambda\\
+
\frac{1}{2}
\sum_{r, t, v ,w = 1}^p
\freqint \freqint 
\Bigl(
\frac{\partial}{\partial \theta_a}\bm{f}_{\btheta}^{rt}(\lambda_1) \cdot
\frac{\partial}{\partial \theta_b}\bm{f}_{\btheta}^{vw}(\lambda_2)
\tilde{\gamma}_{rtvw} (u; -\lambda_1, \lambda_2, -\lambda_2)
\Bigr) \dif \lambda_1 \dif \lambda_2
\Biggr),
\end{multline}
$(a, b = 1, \dots, d)$,
where $\tilde{\gamma}$ is the fourth-order spectral density of the process.
\end{thm}

\subsection{Inference for causality measures} \label{subsec:3.2}
In this subsection,
we develop the statistical inference for the local Granger causality measure \eqref{eq:2.15}
based on the parametric model $\{ \bm{f}_{\bm{\theta}(u)} \mid \bm{\theta}(u) \in \Theta \}$.
Denote by $\Sigma_{\bm{\theta}(u)}$ the parametric one-step-ahead prediction error matrix,
and
by $\bm{\frak{f}}_{\bm{\theta}(u)}$ the parametric model for the companion process \eqref{eq:r2.6}.
Note that the matrix $\bm{\frak{f}}_{\bm{\theta}(u)}$ is uniquely determined by 
the model $\bm{f}_{\bm{\theta}(u)}$ and the matrix $\Sigma_{\bm{\theta}(u)}$ (see, e.g., \eqref{eq:r3.2.6} and \eqref{eq:r4.2.9}).

Suppose $\bm{f}_{\bm{\theta}(u)}$,
$\Sigma_{\bm{\theta}(u)}$
and
$\bm{\frak{f}}_{\bm{\theta}(u)}$
have the same partition as \eqref{eq:2.11}.
To make the statistical inference feasible, we impose the following assumption
on the parametric models.

\begin{asp} \label{asp:3.3}
For any $\btheta \in \Theta$,
\[
\freqint \log \abs{\pf} \dif \lambda > -\infty.
\]
\end{asp}

Assumption \ref{asp:3.3} is usually referred to as the maximal rank condition.
Under Assumption \ref{asp:3.3}, $\pfu$ is non-degenerate a.e.\ and $\Sigma_{\bm{\theta}(u)}$
is positive definite for any fixed $u \in [0, 1]$.
Now the parametric local Granger causality for \eqref{eq:2.15} is
\begin{equation} \label{eq:3.11}
{\rm GC}^{(2 \to 1)}(u; \,\btheta) = \frac{1}{2\pi} \freqint
{\rm FGC}\bigl(\lambda; \btheta(u)\bigr) \dif \lambda,
\end{equation}
where
\begin{equation*}
{\rm FGC}(\lambda; \,\bm{\theta}) =
\log
\frac{
\abs{\bm{f}_{\btheta}(\lambda)_{11} }
}{
{\Babs{\bm{f}_{\btheta}(\lambda)_{11} 
- 2\pi \bm{\frak{f}}_{\btheta}(\lambda)_{12} 
\tilde{\Sigma}_{\btheta, 22}^{-1} 
\bm{\frak{f}}_{\btheta}(\lambda)_{21} }}
}.
\end{equation*}
The main results are described in the following.
\begin{thm} \label{thm:3.7}
Suppose Assumptions \ref{asp:2.1}, \ref{asp:3.1}, \ref{asp:3.2} and \ref{asp:3.3} hold.
If we have, for some $i = 1, \dots, d$,
\begin{equation}\label{eq:3.20}
\frac{\partial}{\partial \theta_i} {\rm FGC}^{(2 \to 1)}_{\bm{\theta}}(\lambda) \Bigr\vert_{\btheta = \btheta_0(u)}
\not = 0,
\qquad \text{for some $\lambda \in (-\pi, \pi]$,}
\end{equation}
uniformly in $u \in [0, 1]$, and if $b_T^{-1} = o\bigl(T(\log T)^{-6}\bigr)$ and $b_T = o(T^{-1/5})$, 
then we have
\[
\sqrt{T b_T}
\Bigl(
{\rm GC}^{(2 \to 1)}(u;\,{\hat{\bm{\theta}}_T}) - {\rm GC}^{(2 \to 1)}(u; {\bm{\theta}_0})
\Bigr) 
\dlim
\mathcal{N}(0, \mathbb{V}^{\text{GC}}(u)),
\]
where
\begin{equation*}
\mathbb{V}^{\text{GC}}(u)
=
\left(\nabla {\rm GC}^{(2 \to 1)}(u; {\bm{\theta}_0})  \right)^{\top}
\mathbb{V}(u) 
\left( \nabla {\rm GC}^{(2 \to 1)}(u; {\bm{\theta}_0}) \right),
\end{equation*}
and
\[
\nabla {\rm GC}^{(2 \to 1)}(u; {\bm{\theta}_0})
= \left( \frac{\partial}{\partial \theta_1} {\rm GC}^{(2 \to 1)}(u; {\bm{\theta}_0}), \dots, 
\frac{\partial}{\partial \theta_d} {\rm GC}^{(2 \to 1)}(u; {\bm{\theta}_0}) \right)^{\top}
\]
with
\begin{equation*} 
\frac{\partial}{\partial \theta_i} {\rm GC}^{(2 \to 1)}(u; {\bm{\theta}})
= \frac{1}{2\pi} \int_{-\pi}^{\pi} \frac{\partial}{\partial \theta_i} {\rm FGC}^{(2 \to 1)}_{\bm{\theta}(u)}(\lambda) \dif \lambda. 
\end{equation*}
\end{thm}

There are situations when condition \eqref{eq:3.20} may not be satisfied.
That is, for some $u \in (0, 1)$, 
\begin{equation} \label{eq:3.23}
\frac{\partial}{\partial \btheta} {\rm FGC}^{(2 \to 1)}_{\bm{\theta}}(\lambda) \Bigr\vert_{\btheta = \btheta_0(u)}
= \bzero,
\qquad \text{a.e. $\lambda \in (-\pi, \pi]$.}
\end{equation}
In this case, we centralize ${\rm GC}^{(2 \to 1)}(u;\, {\hat{\bm{\theta}}_T})$ as ${\rm CGC}(u;\,{\hat{\bm{\theta}}_T} )$, i.e.,
\begin{equation} \label{eq:3.24}
{\rm CGC}(u;\,{\hat{\bm{\theta}}_T} ) := 
{\rm GC}^{(2 \to 1)}(u;\,
{\hat{\bm{\theta}}_T}) - 
{\rm GC}^{(2 \to 1)}(u; {\bm{\theta}_0}).
\end{equation}
Then we have the following result.
\begin{thm} \label{thm:3.8}
Suppose that Assumptions \ref{asp:2.1}, \ref{asp:3.1}, \ref{asp:3.2} and \ref{asp:3.3} hold.
In addition, assume \eqref{eq:3.23} with
\begin{equation*}
\mathcal{H}(u, \lambda): = \frac{\partial^2}{\partial \btheta \partial \btheta^{\top}} 
{\rm FGC}^{(2 \to 1)}_{\bm{\theta}}(\lambda) \Bigr\vert_{\btheta = \btheta_0(u)}
\not = O_{d \times d},
\qquad \text{for some $\lambda \in (-\pi, \pi]$},
\end{equation*}
for $u \in (0, 1)$.
Let 
\begin{equation} \label{eq:r4.3.25}
\mathcal{H}(u) = \frac{1}{2\pi} \freqint \mathcal{H}(u, \lambda) \dif \lambda.
\end{equation}
Then if $b_T^{-1} = o\bigl(T(\log T)^{-6}\bigr)$ and $b_T = o(T^{-1/5})$, the following result holds
\begin{equation} \label{eq:3.25}
2 T b_T {\rm CGC}^{(2 \to 1)}(u;\,{\hat{\bm{\theta}}_T}) 
\dlim
\mathcal{N}\bigl(0, \mathbb{V}(u)\bigr)^{\top} \mathcal{H}(u) \mathcal{N}\bigl(0, \mathbb{V}(u)\bigr),
\end{equation}
where the normal distribution $\mathcal{N}\bigl(0, \mathbb{V}(u)\bigr)$ is defined in Theorem \ref{thm:3.3}.
In particular, if $\mathbb{V}(u)^{-1/2}\mathcal{H}(u) \mathbb{V}(u)^{-1/2}$ is an idempotent matrix, then 
the right hand side of \eqref{eq:3.25}
has a chi-squared distribution $\chi^2_{\nu}$ 
with the degrees of freedom 
\[
\nu = \tr(\mathbb{V}(u)^{-1/2}\mathcal{H}(u)\mathbb{V}(u)^{-1/2}).
\]
\end{thm}

\begin{eg}[Time-varying vector autoregression model] \label{eg:4} 
Suppose the multivariate locally Gaussian stationary process \eqref{eq:2.1} has the time-varying spectral density
\begin{equation*}
\fu = \frac{1}{2\pi} 
\bigl(I_2 + A(u) \exp(\I \lambda)\bigr)^{-1}
\begin{pmatrix}
\sigma_{11} & \sigma_{12} \\
\sigma_{12} & \sigma_{22}
\end{pmatrix}
\bigl(I_2 + A(u)^{\top} \exp(-\I \lambda)\bigr)^{-1},
\end{equation*}
where $A(u) = 
\begin{pmatrix}
\alpha_{11}(u) & \alpha_{12}(u) \\
\alpha_{21}(u) & \alpha_{22}(u)
\end{pmatrix}$
and
$\alpha_{12}(u) \equiv 0$.

We adopt the following parametric spectral density $\pf$ for model fitting.
\begin{equation*} 
\pf = \frac{1}{2\pi} 
\Biggl(I_2 + 
\begin{pmatrix}
a_{11} & a_{12} \\
a_{21} & a_{22}
\end{pmatrix} \exp(\I \lambda)\Biggr)^{-1}
\begin{pmatrix}
s_{11} & s_{12} \\
s_{12} & s_{22}
\end{pmatrix}
\Biggl(I_2 + 
\begin{pmatrix}
a_{11} & a_{12} \\
a_{21} & a_{22}
\end{pmatrix}^{\top} \exp(-\I \lambda)\Biggr)^{-1},
\end{equation*}
where $\btheta = (a_{11}, a_{12}, a_{21}, a_{22}, s_{11}, s_{12}, s_{22})^{\top}$.
From the definition \eqref{eq:3.11}, we have
\begin{equation*}
{\rm GC}^{(2 \to 1)}(u; \, \btheta) = 
\frac{1}{2 \pi} \int_{-\pi}^{\pi} 
-\log 
\Babs{
1 - 2 \pi 
\pff_{12}
\left(\tilde{\Sigma}_{\btheta, 22} \right)^{-1} 
\pff_{21} f_{\btheta}(\lambda)_{11}^{-1}}
\dif \lambda.
\end{equation*}
Thus, $ {\rm FGC}^{(2 \to 1)}_{\bm{\theta}}(\lambda)$ is
\[
{\rm FGC}^{(2 \to 1)}_{\bm{\theta}}(\lambda)
=
-\log 
\Babs{
1 - 2 \pi 
\pff_{12}
\left(\tilde{\Sigma}_{\btheta, 22} \right)^{-1} 
\pff_{21} f_{\btheta}(\lambda)_{11}^{-1}}.
\]
A straightforward computation leads to
\[
\frac{\partial}{\partial \btheta} 
{\rm FGC}^{(2 \to 1)}_{\bm{\theta}}(\lambda) \Bigr\vert_{a_{12} = 0}
= \bzero,
\qquad \text{for any $\lambda \in (-\pi, \pi]$.}
\]
In addition, it holds that
\[
\frac{\partial^2}{\partial a_{12}^2} 
{\rm FGC}^{(2 \to 1)}_{\bm{\theta}}(\lambda) \Bigr\vert_{\btheta = \btheta_0(u)}
= 
\frac{2 (\sigma_{11} \sigma_{22} - \sigma_{12}^2)^2}{\sigma_{11}^4} 
\frac{\abs{1 - \alpha_{11}(u) \exp(\I \lambda)}^2}{\abs{1 - \alpha_{22}(u) \exp(\I \lambda)}^2} > 0;
\]
and for all $\lambda \in (-\pi, \pi]$,
\[
\frac{\partial^2}{\partial \theta_i \partial \theta_j} 
{\rm FGC}^{(2 \to 1)}_{\bm{\theta}}(\lambda) \Bigr\vert_{\btheta = \btheta_0(u)}
= 
0, \quad
\text{for $\theta_i \not = a_{12}$ or $\theta_j \not = a_{12}$}.
\]

Let $\hat{\btheta}_T = \bigl(\hat{\alpha}_{11}(u), \hat{\alpha}_{12}(u), \hat{\alpha}_{21}(u), \hat{\alpha}_{22}(u), 
\hat{\sigma}_{11}, \hat{\sigma}_{12}, \hat{\sigma}_{22}\bigr)^{\top}$
be the local Whittle estimator defined in \eqref{eq:3.1}.
Applying Theorem \ref{thm:3.8}, we obtain
\begin{equation} \label{eq:asymp_dist}
T b_T \, \frac{\sigma_{11}^4 \Bigl(\int_{-1}^1 K(v)^2 \dif v\Bigr)^{-1}}
{\bigl(1 + \alpha_{11}(u)^2 - 2 \alpha_{11}(u) \alpha_{22}(u)\bigr)
\bigl(\sigma_{11} \sigma_{22} - \sigma_{12}^2\bigr)^2} {\rm CGC}(u;\,{\hat{\bm{\theta}}_T} )  
\dlim \chi^2_1,
\end{equation}
since 
\begin{multline*}
\mathcal{H}(u)_{22} =
\frac{1}{2\pi} \freqint
\frac{\partial^2}{\partial a_{12}^2} 
{\rm FGC}^{(2 \to 1)}_{\bm{\theta}}(\lambda) \Bigr\vert_{\btheta = \btheta_0(u)}
\dif \lambda\\
=
\frac{2\bigl(1 + \alpha_{11}(u)^2 - 2 \alpha_{11}(u) \alpha_{22}(u)\bigr)
\bigl(\sigma_{11}\sigma_{22} - \sigma_{12}^2\bigr)^2}
{\bigl(1-\alpha_{22}(u)^2\bigr) \sigma_{11}^4}
\end{multline*}
and
\[
\mathbb{V}(u)_{22} = \bigl(1 - \alpha_{22}(u)^2\bigr)
\int_{-1}^1 K(v)^2 \dif v.
\]
\end{eg}

\subsection{Hypothesis testing for causality measures}
We now address the hypothesis testing problem for the local measures ${\rm GC}^{(2 \to 1)}$.
Suppose that we want to test for local causality at a particular rescaled time $u \in [0, 1]$. 
Define the local hypothesis $H_0^{(2 \to 1)}$ to be
\begin{equation} \label{eq:3.27}
H_0^{(2 \to 1)}: {\rm GC}^{(2 \to 1)}(u) = c.
\end{equation}
We consider two cases of the null hypothesis \eqref{eq:3.27}: (i) 
$c = 0$, and (ii) $c > 0$. 

Let us first consider the case (i) $c = 0$.
%
For any fixed $u \in [0, 1]$, with the shorthand $\btheta = \btheta(u)$, we have
\[
{\rm GC}^{(2 \to 1)}(u; \, \btheta)= 
\frac{1}{2 \pi} \int_{-\pi}^{\pi} 
\log \frac{\abs{\pf_{11}}}
{\Babs{
\pf_{11} -
2 \pi 
\pff_{12}
\left(\tilde{\Sigma}_{\btheta, 22} \right)^{-1} 
\pff_{21} }}
\dif \lambda, 
\]
and thus, ${\rm GC}^{(2 \to 1)}(u; \, \btheta)= 0$ if and only if
\begin{equation} \label{eq:3.28}
\abs{\pf_{11}}
=
\Babs{
\pf_{11} -
2 \pi 
\pff_{12}
\left(\tilde{\Sigma}_{\btheta, 22} \right)^{-1} 
\pff_{21} }.
\end{equation}
Here, $\pf_{11} $ and $2 \pi 
\pff_{12}
\left(\tilde{\Sigma}_{\btheta, 22} \right)^{-1} 
\pff_{21}$
are both Hermitian. Thus, the equality \eqref{eq:3.28} holds if 
\[
2 \pi 
\pff_{12}
\left(\tilde{\Sigma}_{\btheta, 22} \right)^{-1} 
\pff_{21}
=
O_{m \times m}.
\]
It follows that $\pff_{12} = O_{m \times M}$, 
since $\left(\tilde{\Sigma}_{\btheta, 22} \right)^{-1} $ is positive definite.
Since $\pff_{12} = O_{m \times M}$, 
it is straightforward to see that
\[
\frac{\partial}{\partial \btheta} 
{\rm FGC}^{(2 \to 1)}_{\bm{\theta}}(\lambda)
= \bzero,
\]
which is the case we considered in Theorem \ref{thm:3.8}.

Accordingly, for the local hypothesis $H_0^{(2 \to 1)}: {\rm GC}^{(2 \to 1)}(u) = 0$, we take 
\[
S^{\dagger}(u): = 
2 T \, b_T \,
{\rm GC}^{(2 \to 1)}(u;\,{\hat{\bm{\theta}}_T})
\]
as the test statistic. Then we have the following result.
\begin{thm}
Suppose Assumptions \ref{asp:2.1}, \ref{asp:3.1}, \ref{asp:3.2} and \ref{asp:3.3} hold.
Under the null hypothesis $H_0^{(2 \to 1)}: {\rm GC}^{(2 \to 1)}(u) = 0$, 
if $b_T^{-1} = o\bigl(T(\log T)^{-6}\bigr)$ and $b_T = o(T^{-1/5})$, it holds that
\begin{equation*}
S^{\dagger}(u)
\dlim
\mathcal{N}\bigl(0, \mathbb{V}(u)\bigr)^{\top} \mathcal{H}(u) \mathcal{N}\bigl(0, \mathbb{V}(u)\bigr).
\end{equation*}
where $\mathcal{N}\bigl(0, \mathbb{V}(u)\bigr)$ is defined in Theorem \ref{thm:3.3}.
\end{thm}

\begin{rem} \label{rem:3.10}
The matrix $\mathcal{H}(u)$ in \eqref{eq:r4.3.25} is unknown in general, 
but it is determinable from the parameterization of $\pf$.
In practice, the matrix $\mathcal{H}(u)$ should be replaced with its plug-in version 
\[
\mathcal{H}(u; \, \hat{\btheta}_T)
:=
\frac{1}{2\pi} \freqint 
\frac{\partial^2}{\partial \btheta \partial \btheta^{\top}} 
{\rm FGC}^{(2 \to 1)}_{\bm{\theta}}(\lambda) \Bigr\vert_{\btheta = \hat{\btheta}_T(u)}
\dif \lambda
\]
to construct an asymptotic $(1- \alpha)$ confidence interval for ${\rm GC}^{(2 \to 1)}(u)$.
Instead, in some situation,
for instance, as a continuation of Example \ref{eg:4}, if we take
\[
\tilde{S}^{\dagger}(u): = 
T \, b_T \,
\frac{\hat{\sigma}_{11}^4 \Bigl(\int_{-1}^1 K(v)^2 \dif v\Bigr)^{-1}}
{\bigl(1 + \hat{\alpha}_{11}(u)^2 - 2 \hat{\alpha}_{11}(u) \hat{\alpha}_{22}(u)\bigr)
\bigl(\hat{\sigma}_{11} \hat{\sigma}_{22} - \hat{\sigma}_{12}^2\bigr)^2} 
{\rm CGC}^{(2 \to 1)}(u;\,{\hat{\bm{\theta}}_T})
\]
as a sample version of the left hand side in \eqref{eq:asymp_dist},
then the confidence interval is $[0, \chi^2_{1, 1 - \alpha}]$, where 
$\chi^2_{1, 1 - \alpha}$ denotes the $(1- \alpha)$ quantile of the chi-squared distribution 
with 1 degree of freedom.
It is straightforward to see that the test by $\tilde{S}^{\dagger}(u)$ is consistent.
If the null hypothesis is rejected, 
then the conclusion is that at level $\alpha$ there is sufficient evidence to suggest that 
there exists local Granger causality from one series to another at rescaled time $u \in [0, 1]$.
\end{rem}

We now move on to the second case (ii) $c > 0$. 
The Wald type test statistic $S^*(u)$ is
\begin{multline} \label{eq:3.29}
S^*(u) =  Tb_T\Bigl({\rm GC}^{(2 \to 1)}(u;\,{\hat{\bm{\theta}}_T}) - c \Bigr)\\
\qquad \qquad \quad
\left[ 
\left(\nabla {\rm GC}^{(2 \to 1)}(u;\,{\hat{\bm{\theta}}_T})  \right)
\mathbb{V}(u) 
\left( \nabla {\rm GC}^{(2 \to 1)}(u;\,{\hat{\bm{\theta}}_T}) \right)^{\top}
\right]^{-1}
\Bigl({\rm GC}^{(2 \to 1)}(u;\,{\hat{\bm{\theta}}_T}) - c \Bigr).
\end{multline}
The following result is a direct consequence of Theorem \ref{thm:3.7}.

\begin{thm}
Suppose Assumptions \ref{asp:2.1}, \ref{asp:3.1}, \ref{asp:3.2} and \ref{asp:3.3} hold.
Under the null hypothesis $H_0^{(2 \to 1)}: {\rm GC}^{(2 \to 1)}(u) = c  > 0$, 
if $b_T^{-1} = o\bigl(T(\log T)^{-6}\bigr)$ and $b_T = o(T^{-1/5})$, we have
\begin{equation*}
S^*(u) \dlim \chi^2_d,
\end{equation*}
where $\chi^2_d$ is a chi-squared distribution with the degrees of freedom $d$.
\end{thm}

\begin{rem}
The covariance matrix $\mathbb{V}(u)$ in \eqref{eq:3.29} is usually unknown,
and thus, we have to construct a consistent estimator $\hat{\mathbb{V}}(u)$ instead.
This can be done by following \cite{keenan1987} or \cite{taniguchi1982}.
\end{rem} 

\section{Numerical Simulations} \label{sec:4}
In this section, we investigate the finite sample performance 
of our proposed local Whittle estimation and the hypothesis testing for non-causality
of multiple time series.

\subsection{Finite sample performance of local Whittle estimation}
Let $\{\xt\}$ be a multivariate locally stationary process defined by
\begin{equation} \label{eq:4.1}
\xt = A\Bigl(\frac{t}{T}\Bigr) \bm{X}_{t-1, T} + \bm{\epsilon}_t,
\qquad
\bm{\ep}_t \sim \mathcal{N}(\bm{0}, I_2).
\end{equation}
We consider two examples for the time-varying coefficient matrix $A(\cdot)$.
Let $A^{({\rm i})}(u)$ and $A^{({\rm ii})}(u)$ be 
\[
A^{({\rm i})}(u) = 
\begin{pmatrix}
1/2 & a_{12}(u) \\
0 & 1/2\\
\end{pmatrix}, \quad
A^{({\rm ii})}(u)
=
\begin{pmatrix}
7/10 & a_{12}(u) \\
0 & 3/10\\
\end{pmatrix},
\]
where $a_{12}(u)$ is defined as
\begin{equation} \label{eq:n4.32}
a_{12}(u)
=
\begin{cases}
0, & \quad \text{$0 \leq u \leq 1/\pi$},\\
\pi (u - 1/\pi)/2, & \quad \text{$1/\pi \leq u \leq 2/\pi$}, \\
1/2, & \quad  \text{$2/\pi \leq u \leq 1$}.
\end{cases}
\end{equation}
Note that $a_{12}$ is a piecewise continuous function.
Additionally, $a_{12}$ is not differentiable at $1/\pi$ and $2/\pi$,
but $a_{12}(t/T)$ does not take values on those points since they are irrational.
Thus, $a_{12}$ could be replaced by a differentiable function instead.

To study the finite sample performance of the local Whittle estimation,
define $\bm{f}_{\theta}^{({\rm i})}$ and $\bm{f}_{\theta}^{({\rm ii})}$ to be 
\[
\bm{f}_{\theta}^{(\bullet)}(\lambda) = 
\frac{1}{2 \pi} A_{\theta}^{(\bullet)}(\lambda)^{-1} \Bigl(A_{\theta}^{(\bullet)}(-\lambda)^{\T}\Bigr)^{-1},
\qquad
\bullet = \text{(i) or (ii)},
\]
where $A_{\theta}^{({\rm i})}(u, \lambda)$ and $A_{\theta}^{({\rm ii})}(u, \lambda)$ are
\[
A_{\theta}^{({\rm i})}(u, \lambda)
=
I_2 - 
\begin{pmatrix}
1/2 & \theta \\
0 & 1/2\\
\end{pmatrix} \exp(-\I \lambda), \quad
A_{\theta}^{({\rm ii})}(u, \lambda)
=
I_2 - 
\begin{pmatrix}
7/10 & \theta \\
0 & 3/10\\
\end{pmatrix} \exp(-\I \lambda).
\]
The local Whittle likelihood estimator $\hat{\theta}_T(u)$ is defined as in \eqref{eq:3.1},
where the Epanechnikov kernel is used with the bandwidth $b_T = 4 T^{1/3}$.
The procedure was examined over 100 simulations for $T = 50$ and $T = 100$, respectively.

Define $\hat{\theta}_T^{(j)}(u)$ to be the estimate from the $j$th dataset where 
$j = 1, \dots, 100$.
The numerical results 
of the mean, 5th and 95th percentile of the estimates $\{\hat{\theta}_T^{(j)}(u)\}$
are shown in Figure \ref{fig:4.1}.
Both results for Examples (i) and (ii) are similar.
In this finite sample performance, 
we see that the parameter estimation loses some accuracy when $u$ is close to edges of $[0, 1]$.
On the other hand, the estimate $\hat{\theta}_T(u)$ works well around the center of the interval. 
The estimates capture the general form of the true function well.
We can expect that the performance could be better if the change in the function $a_{12}(u)$ is smaller.
Here, we remark that the estimation is asymptotically unbiased from the theoretical results.
Compared with the case when $T = 50$,
the confidence interval has narrower widths with the same coverage probability when $T = 100$.
These observations justify our theoretical results.

\begin{figure}
\begin{center}
 \subfigure 
 {\includegraphics[clip, width=0.3\columnwidth]{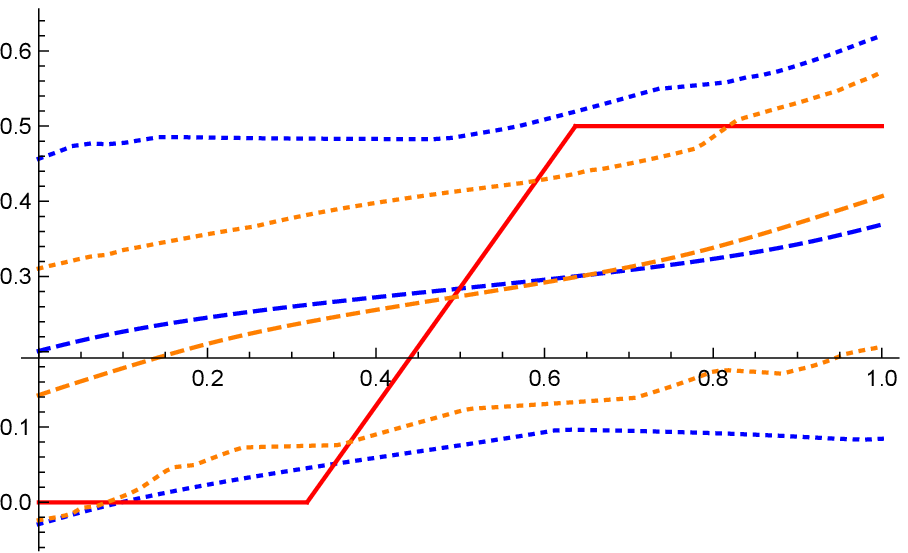}}  \qquad  \qquad 
 \subfigure 
 {\includegraphics[clip, width=0.3\columnwidth]{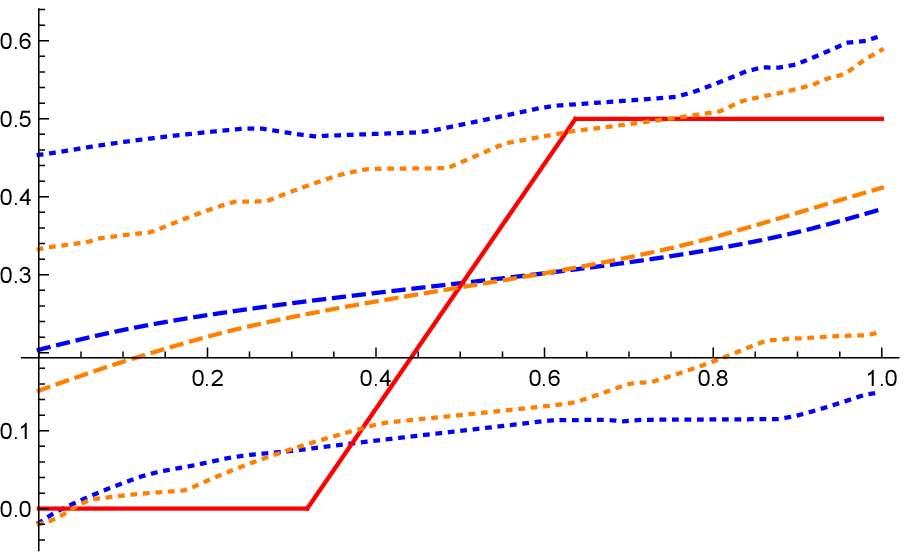}} 
 \caption{
 Example (i) (left) and Example (ii) (right). The function $a_{12}$ is shown in red;
 The mean, 5th and 95th percentile of the estimates when $T = 50$ are shown in blue;
 The mean, 5th and 95th percentile of the estimates when $T = 100$ are shown in orange.}
 \label{fig:4.1}
\end{center}
\end{figure}

\subsection{Testing for non-causality}
Let us consider the model \eqref{eq:4.1} again.
We replace the function $a_{12}(u)$ in \eqref{eq:n4.32} by
\begin{equation} \label{eq:4.2}
a_{12}(u)
=
\begin{cases}
0, & \quad \text{$0 \leq u \leq \frac{5}{2\pi}$},\\
\frac{\pi}{2} (u - \frac{5}{2\pi}), & \quad \text{$\frac{5}{2\pi} \leq u \leq 1$}.
\end{cases}
\end{equation}
In this simulation study, we tested the local hypothesis
\begin{equation} \label{eq:4.3}
H_0^{(2 \to 1)}: GC^{(2 \to 1)}(u) = 0.
\end{equation}
for the following values of rescaled time $u = 0.1, 0.3, 0.5, 0.7, 0.9$.
The null hypothesis is rejected whenever
\begin{equation} \label{eq:4.4}
\tilde{S}^{\dagger}(u) > \chi_{1, 1 - \alpha}^2,
\end{equation}
which the details for have been considered in Remark \ref{rem:3.10}.

The rejection probabilities of the test \eqref{eq:4.4}
based on 1000 simulations are reported in Table 1. 
The rejection probabilities when $u = 0.1$ is very close to the nominal significance level $\alpha$.
However, the rejection probabilities when $u = 0.3, 0.5, 0.7$
seem conservative compared with the nominal significance level.
This stems from the different asymptotic distributions between $\alpha_{12} = 0$
and $\alpha_{12} \not = 0$ (see Theorems \ref{thm:3.7} and \ref{thm:3.8}).
Table 1 
confirms that the local power for testing the null hypothesis \eqref{eq:4.3} increase with $u$,
because of the feature of the function $a_{12}$ in the model \eqref{eq:4.1}.
In summary, the testing for non-local Granger causality
can be seen as a methodology to find a predictive sign 
to detect the structural change of a dynamic model.

\begin{table} 
\label{tbl:n4.1} 
\caption{Rejection probabilities of the test \eqref{eq:4.4} for the model \eqref{eq:4.1}
with $a_{12}$ in \eqref{eq:4.2}.}
\begin{center}
\begin{tabular}{cccccc} \hline
$\alpha$ & $u = 0.1$ & $u = 0.3$ & $u = 0.5$ & $u = 0.7$ & $u = 0.9$ \\ \hline
1\% & 0.007 & 0.000 & 0.002 & 0.002 & 0.019 \\  
5\% & 0.043 & 0.020 & 0.015 & 0.030 & 0.065 \\ 
10\% & 0.096 & 0.053 & 0.037 & 0.081 & 0.123 \\
15\% & 0.152 & 0.085 & 0.071 & 0.111 & 0.186 \\ \hline
\end{tabular}
\end{center}
\end{table}%

\section{Data Analysis} \label{sec:5}
In this section, we apply local Granger causality to two real datasets -- EEG data and financial data.

\subsection{EEG data}
We provide a brief description of the data.
The EEG signals are sampled at the rate of 100 Hertz. The recordings are taken from channels 
the central channels (C3, C4, Cz), parietal channels (P3, P4) and the temporal channels (T3, T4, T5)
which correspond roughly to the central, parietal and temporal brain cortical regions (See Figure \ref{fig:5.1}). 
The original dataset has 32680 time points for each channel (i.e., the period of the observation is 
326.8 seconds). This dataset was previously analyzed in \cite{osg2005} and \cite{so2019}. However, 
none of these two papers addressed the very important issue of causality. This is the first paper that 
examined local Granger causality features in this data.

\begin{figure}[H]
\begin{center}
{\includegraphics[clip, width=0.3\columnwidth]{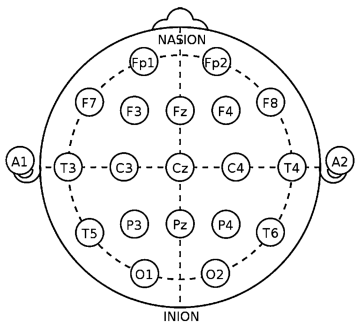}}  
 \caption{EEG channels.} 
 \label{fig:5.1}
\end{center}
\end{figure}

Local Granger causality was estimated and tested at 
every rescaled time point $u_k = 2.1 k/326.8$ and $u_k = 4.2 k/326.8$,
due to the computational cost of the local Whittle estimation.
This is equivalent to estimating and testing every $2.1$ and $4.2$ seconds respectively.
We refer to these partial data as one at regular intervals of 2.1 seconds and 4.2 seconds.
In Figure \ref{fig:5.6}, 
we show the logarithm of
local Granger causality between two specific channels P3 (left parietal) and T3 (left temporal). 
These two channels are of primary interest because the patient suffered from left temporal lobe epilepsy - though the precise location is quite close to the parietal lobe. 
Thus the seizure focus is the left temporal lobe and any abnormalities in the EEG are captured in 
the T3 and P3 channels. 
The 95\% confidence intervals are shown below: 
the dashed one is computed from the data at regular intervals of 4.2 seconds;
the dotted one is computed from the data at regular intervals of 2.1 seconds.

\begin{figure}[H]
\center
 \subfigure 
 {\includegraphics[clip, width=0.35\columnwidth]{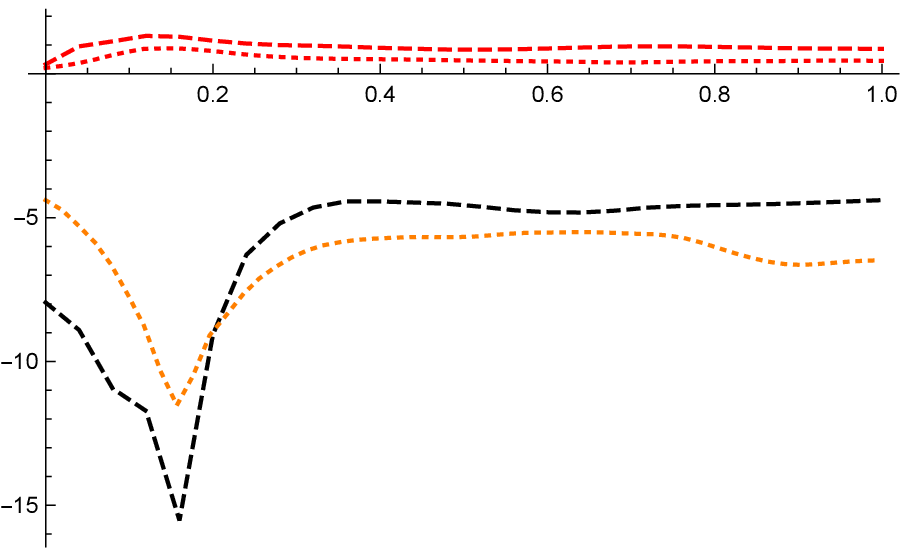}}  \qquad  \qquad 
 \subfigure 
 {\includegraphics[clip, width=0.4\columnwidth]{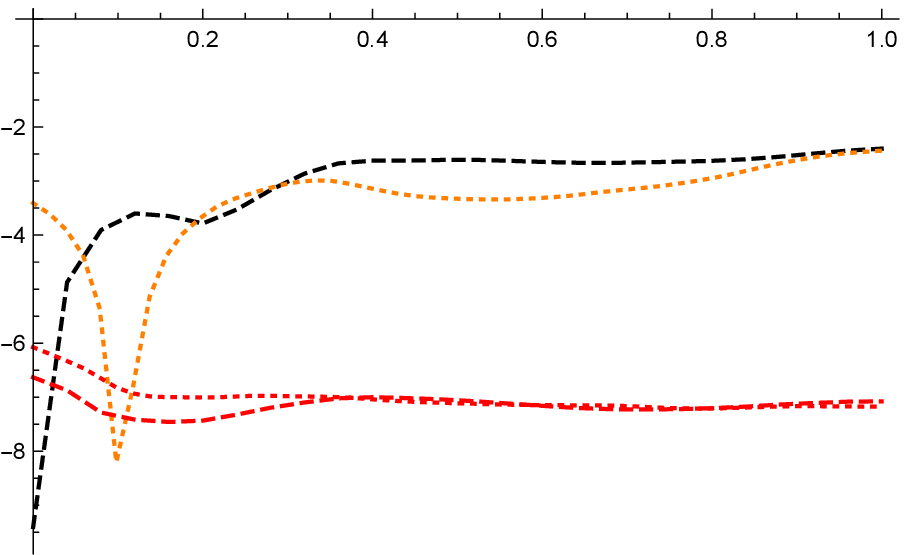}}  
 \caption{The logarithm of local Granger causality (LLGC) from the channel T3 to P3 (left) and that from the channel P3 to T3 (right) during the rescaled time $u \in [0, 1]$.
 The LLGC for data at regular intervals of 4.2 seconds is shown in black and 
 that for those of 2.1 seconds are shown in orange. 
 The dashed red line shows the 95\% confidence levels computed from the data 
 at regular intervals of 4.2 seconds,
 while the dotted red line shows the 95\% confidence levels computed from the data 
 at regular intervals of 2.1 seconds. } 
 \label{fig:5.6}
\end{figure}

The partial data at regular intervals of 2.1 seconds and 4.2 seconds share a very similar move of 
the logarithm of local Granger causality and the 95\% confidence intervals are also very similar.
In general, the higher temporal resolution the sampling is,
theoretically more accurate the estimates are.
In our simulation results,
the estimates from the data at regular intervals of 2.1 seconds are as good as that of 4.2 seconds.
The analysis suggests that T3 does not cause P3 in the Granger sense,
but P3 causes T3 in the Granger sense at latest after the rescaled time $u = 0.15$. 
This is a quite interesting finding because previous analyses have focused on the T3 channel because of the 
distinctly large amplitudes immediately post-seizure onset. However, the novel finding here is that 
the direction of Granger causality actually flows from P3 to T3. This suggests that, despite the 
relatively lower amplitude changes in P3, it still explains the future large amplitude fluctuations in 
T3. 

\begin{figure}[H]
\begin{flushright}
{\includegraphics[width= 0.9\columnwidth]{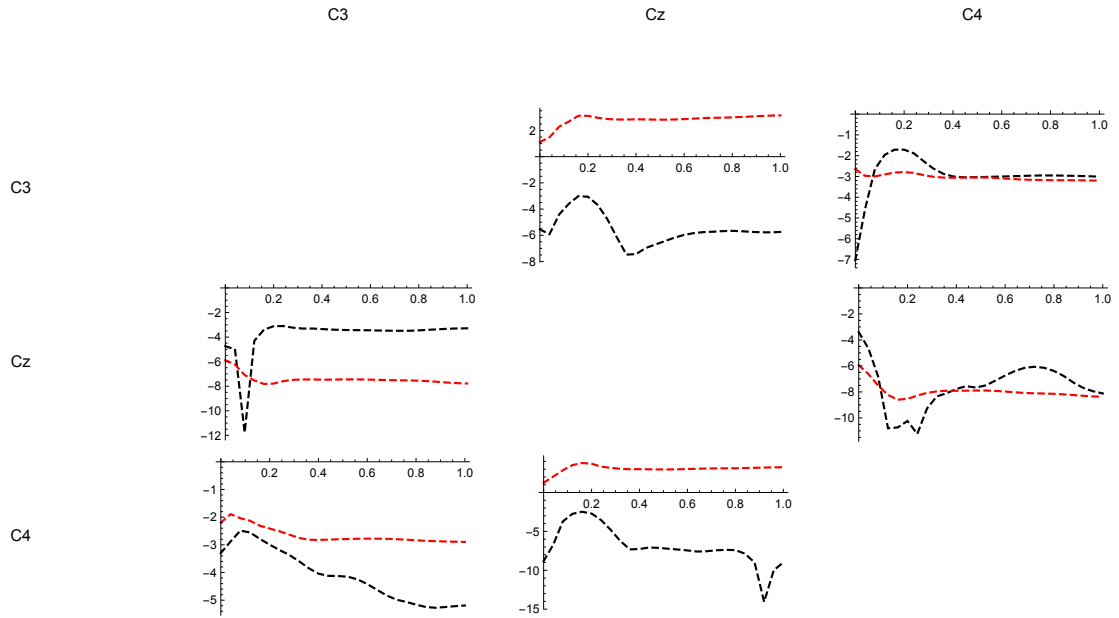}}  
\end{flushright}
 \caption{Plots of the logarithms of local Granger causality from the channels C3, Cz, C4 in the column
 to the channels  C3, Cz, C4  in the row.
The 95\% confidence intervals are below the dashed red lines;
The dashed black lines show the logarithms of local Granger causalities.} 
 \label{fig:5.7}
\end{figure}

Next, we further investigate the local causalities between the central channels C3, Cz and C4 at regular intervals of 
4.2 seconds. Figure \ref{fig:5.7} represents the numerical results of the causality
from the column to the row.
For example, the middle plot in the first row shows the causality from the channel C3 to Cz.
Still, the 95\% confidence intervals are below the dashed red lines
and the logarithms of local Granger causalities are shown by the dashed black lines.
From Figure  \ref{fig:5.7}, the conclusion is that channel the left central channel C3 does not cause 
central channel Cz. Moreover, the right central channel C4 does not cause channels C3 and Cz 
- in the local Granger sense.
In other cases, local Granger causality changes across the evolution of the epileptic seizure 
which confirms the dynamic activity of the brain.
This is a new finding since all previous analyses were limited to modeling
dependence using only coherence which accounts for contemporaneous dependence;
that is, there was no phase or lead-lag analysis. Moreover, this novel finding is quite interesting 
because a change in the causality structure was captured even before the onset of the epileptic seizure, 
which was approximately at $u=0.5$.

\begin{figure}[H]
\begin{flushright}
{\includegraphics[width=0.9\columnwidth]{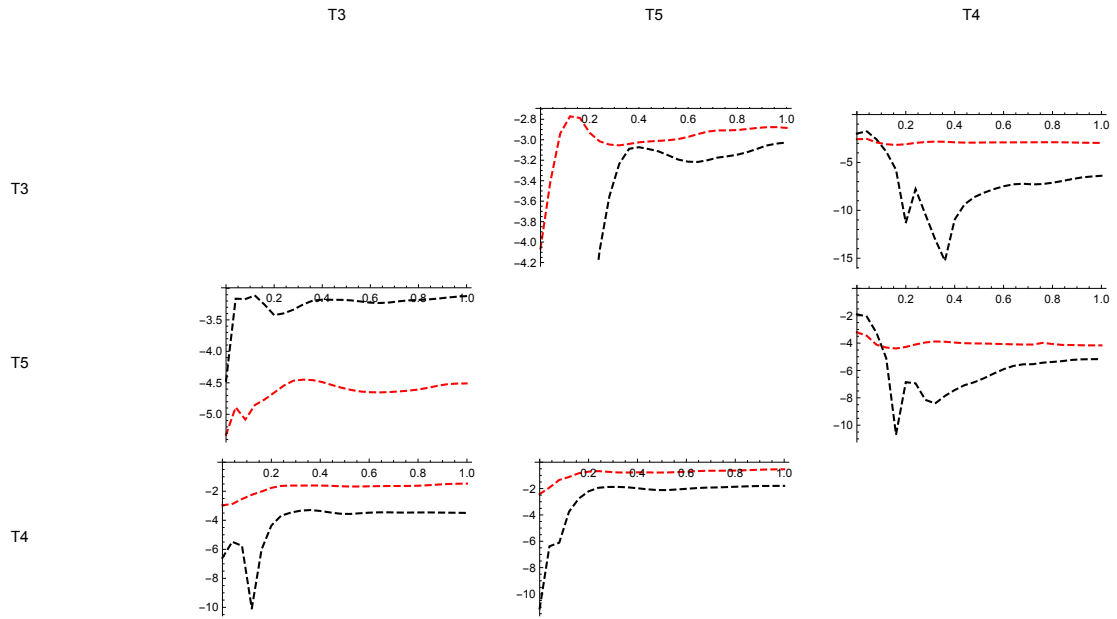}} 
\end{flushright}
 \caption{Plots of the logarithms of local Granger causality from the channels  T3, T5, T4 in the column
 to the channels  T3, T5, T4 in the row.
The 95\% confidence intervals are below the dashed red lines;
The dashed black lines show the logarithms of local Granger causalities.} 
 \label{fig:5.8}
\end{figure}

Similarly, we studied the causalities between the temporal channels T3, T5, T4 at regular intervals of 4.2 seconds.
The plots of the numerical results are shown in Figure \ref{fig:5.8}.
Remember that the channels T3 and T4 are symmetrically located at both the left and right temporal 
cortical regions, respectively. The plots suggest that T3 and T4 do not cause each other in the local 
Granger sense.
Furthermore, the channel T5, also on the left temporal cortical region, uniformly causes T3 in the data which is another interesting novel finding.
It is already known to the neurologist that the patient
has left temporal lobe epilepsy 
and that seizure events are generally initiated in the "left temporal"l region
(which is the area covered by the T3 and T5 channels).
Using the novel proposed concept of local Granger causality,
the analysis produced a highly specific result of brain functional connectivity,
that is, the direction goes from T5 to T3 and not the other way around.
As an additional result,
P3 and P4 do not cause each other uniformly in the local Granger sense at 95\% confidence level.

\subsection{Stock market data}
The dataset is the weekly log-returns 
of the closing stock prices of two financial groups (Mitsubishi and Mizuho) in the Nikkei index.
For brevity, two financial groups are denoted by A and B here.
The weekly data are from 2006 January 1st to 2010 December 26th,
so the length of the data is $T = 260$.

In our data analysis, we compute local Granger causality 
${\rm GC}^{{\rm A} \to {\rm B}}(u)$ and ${\rm GC}^{{\rm B} \to {\rm A}}(u)$
for rescaled time $u = 1/T, 2/T, \dots, 1$.
In general, Granger causality is not symmetric (e.g., the analysis of EEG data) and 
we regard
\begin{equation} \label{eq:5.1}
\biggl\{
\Bigl(
{\rm GC}^{{\rm A} \to {\rm B}}(u), {\rm GC}^{{\rm B} \to {\rm A}}(u)
\Bigr) \in \R^2
\biggr\}_{u = \frac{1}{T}, \frac{2}{T}\dots, 1} 
\end{equation}
as a point cloud in $\R^2$.
We separate local Granger causality in \eqref{eq:5.1} into two parts:
(1) from 2006 January 1st to 2008 June 29th;
(2) from 2008 July 6th to 2010 December 26th.
The part (1) and part (2) have the same length, i.e., the length of each part is 130.

We apply computational topology tools, {\it persistence diagram}, {\it persistence barcode},
{\it persistence landscape},
to capture the feature of these data points of local Granger causality 
(See, e.g.~\cite{flrwb2014}, for the details of the persistence diagram and persistence barcode).
The plots of the persistence diagram and the persistence barcode are shown in Figure \ref{fig:5.9}.
The permutation test is applied to the point clouds (1) and (2)
to test for equal topologies of the point clouds (1) and (2).
In other words, the null hypothesis is no statistical difference between 
the persistence landscapes of local Granger causality.
The result is statistically significant at the significance level of 0.01.
It is known that there is a financial crisis between 2007 and 2008.
Through the analysis of local Granger causality,
we detected the structural change of the causality between these two financial data.
A theoretical justification of this approach will be left as future work.

\begin{figure}[H]
\begin{center}
 \subfigure 
 {\includegraphics[clip, width=0.35\columnwidth]{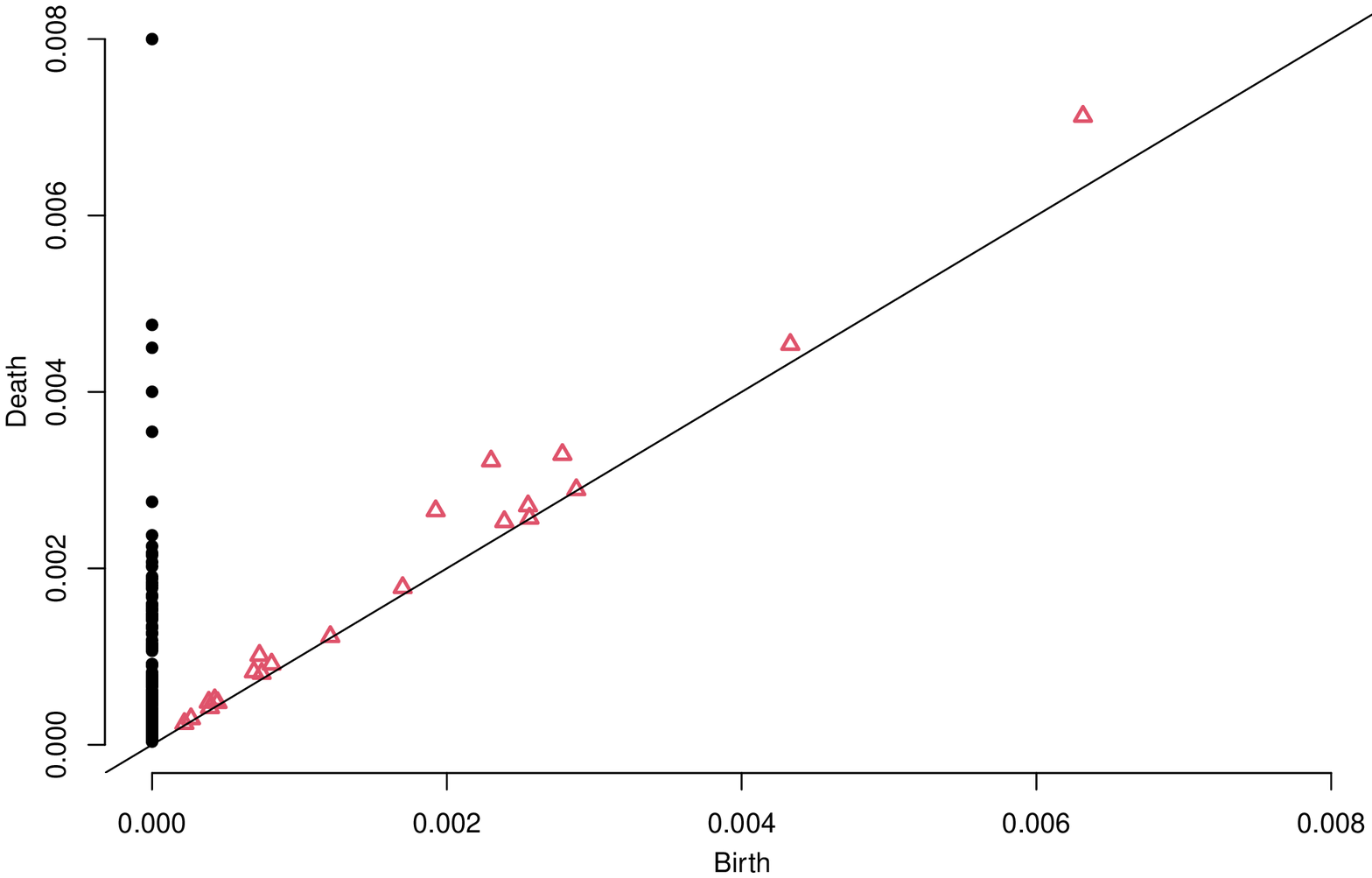}}  \qquad  \qquad
 \subfigure 
 {\includegraphics[clip, width=0.4\columnwidth]{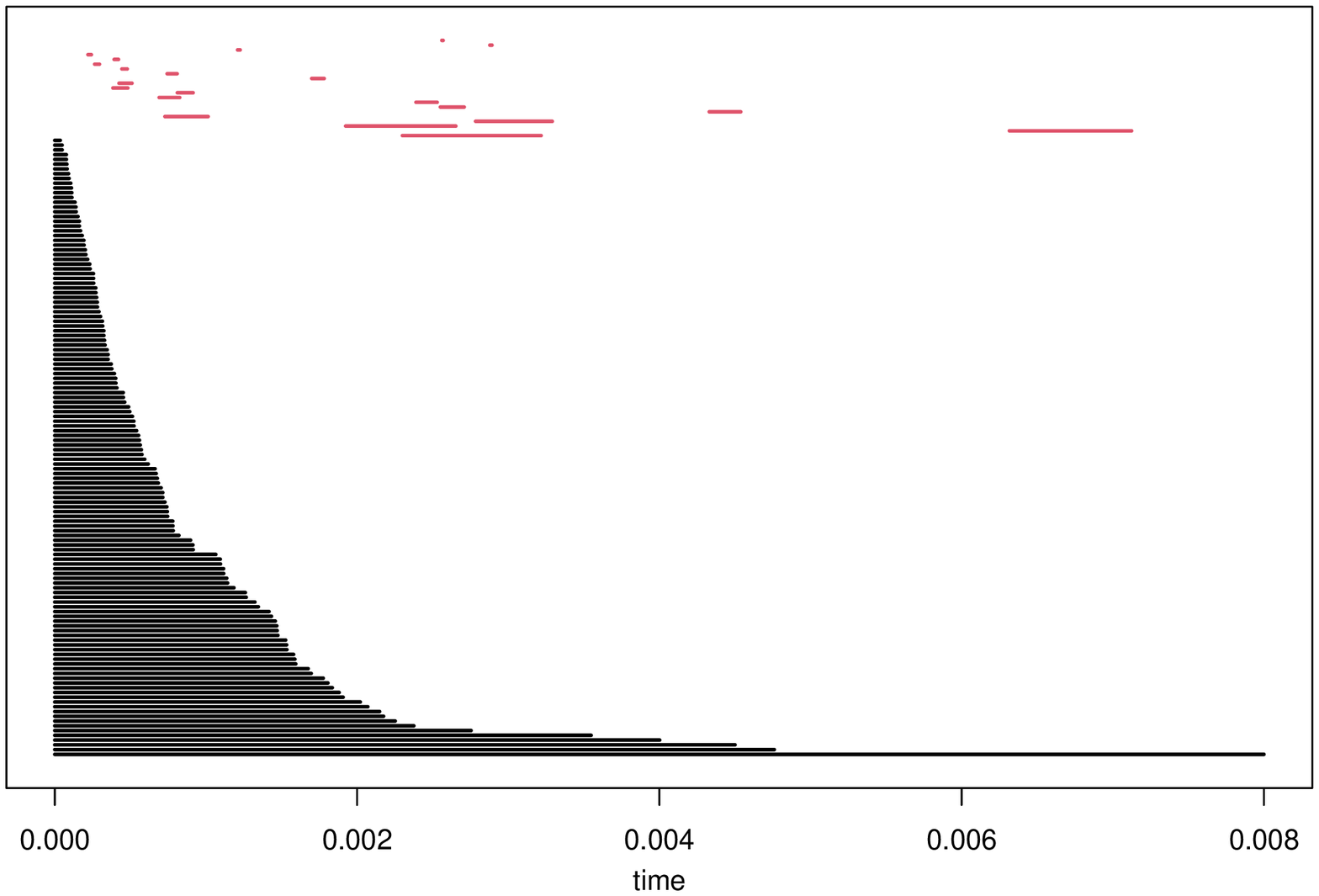}} \\
 \subfigure 
 {\includegraphics[clip, width=0.35\columnwidth]{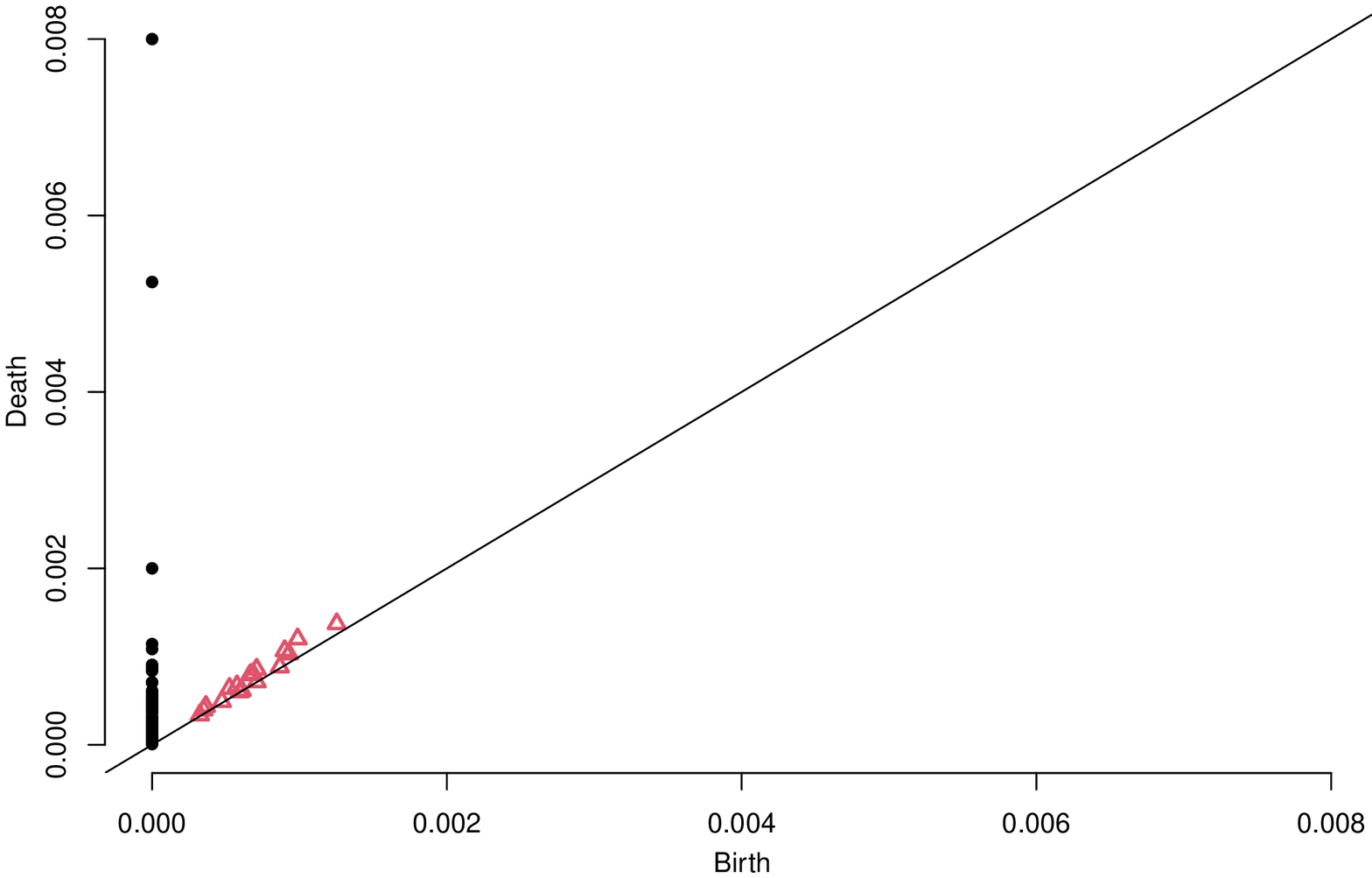}}  \qquad  \qquad
 \subfigure 
 {\includegraphics[clip, width=0.4\columnwidth]{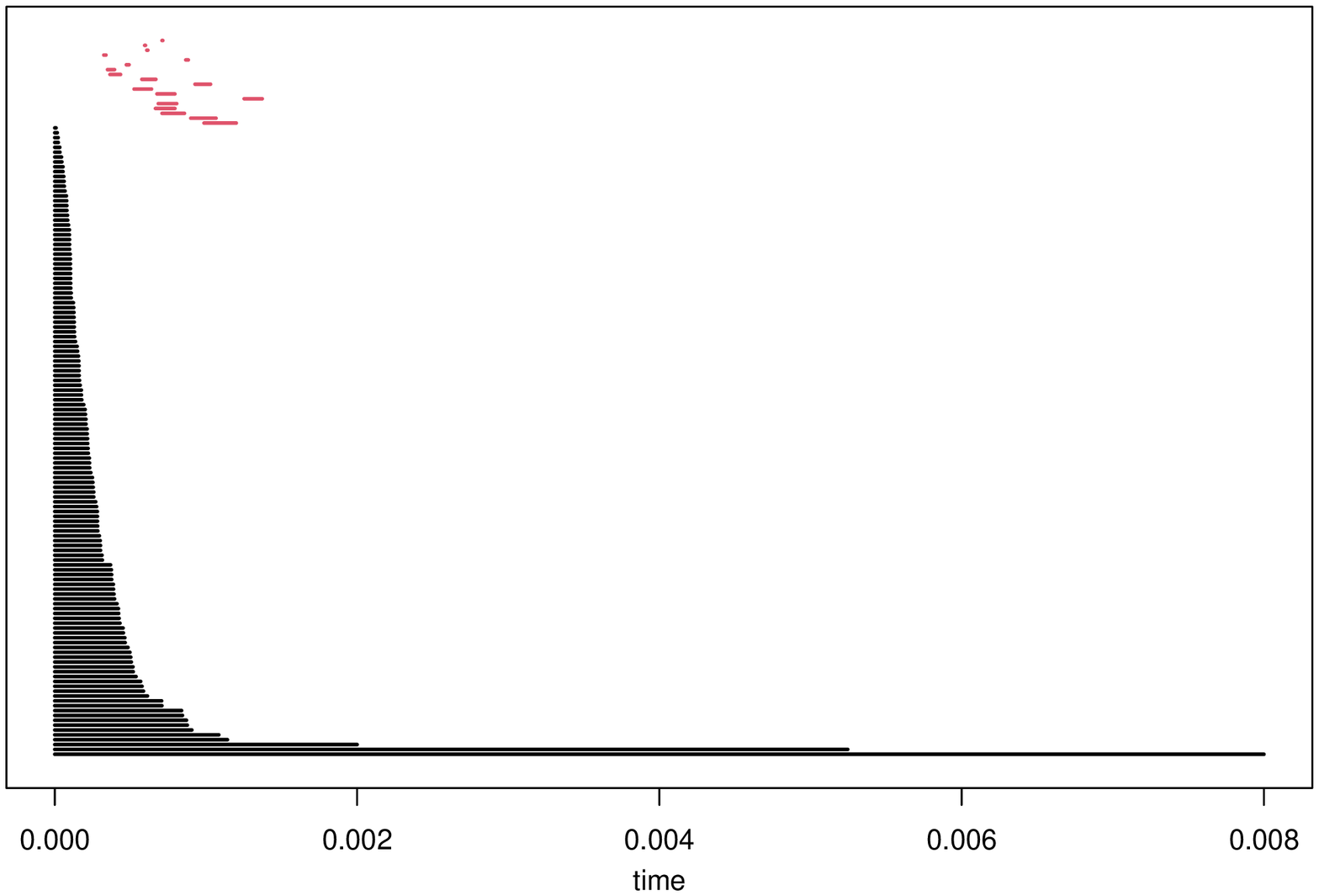}} 
 \end{center}
 \caption{(Above) the persistence diagram of local Granger causality (left);
 and the persistence barcode of local Granger causality (right) during the period (1).
(Below) the persistence diagram of local Granger causality (left);
 and the persistence barcode of local Granger causality (right) during the period (2). } 
 \label{fig:5.9}
\end{figure}

\section{Conclusion}
The primary contribution of this paper is statistical inference for local Granger causality for multivariate time series 
under the framework of multivariate locally stationary processes. 
Our proposed concept of local Granger causality is a generalization of Geweke's measure and Hosoya's measure 
- both of which were developed only for stationary processes. 
Our proposed generalization is well characterized in the frequency domain and has the advantage of being able to capture time-evolving causality relationships. 
We 
developed a procedure for hypothesis testing for the existence of the local Granger causality
from a parametric point of view.
We demonstrate, through the analysis of real data,  the efficiency and efficacy of this procedure to
find the time-evolving aspects of the local Granger causality, which could be overlooked by the existing method
for causality analysis.

In summary, we proposed a consistent method to detect the time change of the local Granger causality.
While our proposed method is nonparametric, 
we note that a procedure for stationary processes was developed in \cite{tpk1996}. 
This could serve as an inspiration for constructing a nonparametric method for locally stationary processes
to test for the local Granger causality, and compare the performance of both approaches.
For multiple time series, to investigate the time change of the local Granger causality
also suffers from the curse of dimensionality.
There are many remaining challenges including dimension reduction in terms of causality
between each component of multiple time series.
In addition to the Lasso method in \cite{tibshirani1996}, most penalized estimation procedures
could be added to our parametric approach to shrink some minor causality between components.
A frequency-specific local causality approach will also be elucidated in our future work.

\section*{Acknowledgements.}
The first author gratefully acknowledge JSPS Grant-in-Aid for Young Scientists (B) 17K12652
and JSPS Grant-in-Aid for Scientific Research (C) 20K11719.
The second author gratefully acknowledge JSPS Grant-in-Aid for Scientific Research (S) 18H05290
The third author gratefully acknowledge the KAUST Research Fund.
The first two authors also would like to express their thanks to the Institute for Mathematical Science (IMS),
Waseda University, for their support for this research.
 

\bibliographystyle{econas} 
\bibliography{clls}

\vspace{-0.5cm}

\newpage
\appendix

\section{Proofs} \label{sec:6}
In Section \ref{sec:6}, we provide proofs for results in Section \ref{sec:3}.
The fundamental properties of multivariate locally stationary processes are condensed in Section \ref{app:a}.
Some technical results to derive the asymptotic distributions are summarized in Section \ref{app:b}.

\subsection{Proof of Theorem \ref{thm:3.2}}
\begin{proof}
From equations \eqref{eq:r15} and \eqref{eq:r17}, we have
\begin{multline*}
\mathcal{L}_T(\btheta, u) - \mathcal{L}(\btheta, u)
=  \\
\biggl\{
\frac{1}{T} \sum_{k=1}^T
\frac{1}{b_T}
K\Bigl(
\frac{u - u_k}{b_T}
\biggr)
-
1
\Bigr\}
\freqint  \log \det \pf \dif\lambda
+
\tr \Bigl(\fu \spf{-1}
\Bigr) \dif \lambda
\\
+
\frac{1}{T} \sum_{k=1}^T
\frac{1}{b_T}
K\Bigl(
\frac{u - u_k}{b_T}
\Bigr) 
\freqint 
\tr \Bigl\{ \Bigl(\fu[u_k] - \fu \Bigr)
\spf{-1}
\Bigr\}
\dif \lambda
\\
+
\frac{1}{T} \sum_{k=1}^T
\frac{1}{b_T}
K\Bigl(
\frac{u - u_k}{b_T}
\Bigr) 
\freqint 
\tr \Bigl\{ \Bigl(\bI_{T}(u_k, \lambda) - \fu[u_k] \Bigr)
\spf{-1}
\Bigr\}
\dif \lambda
\\
=
L_1 + L_2 + L_3, \quad \text{(say).}
\end{multline*}
Since $K$ is a function of bounded variation, applying Lemma P5.1 in \cite{brillinger1981},
it holds that
\begin{align*}
\frac{1}{T} \sum_{k=1}^T
\frac{1}{b_T}
K\Bigl(
\frac{u - u_k}{b_T}
\biggr)
-
1 
&= 
\uint \frac{1}{b_T}
K\Bigl(
\frac{u - v}{b_T}
\biggr)
\dif v
-
1
+
O(T^{-1})\\
&=
\int_{\frac{u - 1}{b_T}}^{\frac{u}{b_T}} K(x) \dif x - 1 + O(T^{-1}),
\end{align*}
which implies that $L_1 = O(T^{-1})$, since the kernel $K$ has a compact support.

Under Assumption \ref{asp:3.1} (i), $\bm{f}$ is of bounded variation, and again, 
applying Lemma P5.1 in \cite{brillinger1981}, we have
\begin{align*}
&\frac{1}{T} \sum_{k=1}^T
\frac{1}{b_T}
K\Bigl(
\frac{u - u_k}{b_T}
\biggr) \Bigl( \fu[u_k] - \fu\Bigr)\\
= &
\uint \frac{1}{b_T}
K\Bigl(
\frac{u - v}{b_T}
\biggr)
\bigl(\fu[v] - \fu\bigr)
\dif v + O(T^{-1}) \\
= &
\int_{\frac{u - 1}{b_T}}^{\frac{u}{b_T}} K(x)
\bigl(\fu[u - b_T x] - \fu\bigr)
\dif x + O(T^{-1}), \\
= &
\int_{\frac{u - 1}{b_T}}^{\frac{u}{b_T}} K(x)
\biggl(
- b_T x \frac{\partial}{\partial u}\fu
+ (b_T x)^2\frac{\partial^2}{\partial u^2}\fu
+ O(b_T^3)
\biggr)
\dif x + O(T^{-1}).
\end{align*}
Since $K$ has a compact support and it is symmetric, we have
\[
\int_{-\infty}^{\infty} x K(x) \dif x = 0,
\]
which implies that 
\[
L_2 = O(b_T^2) + O(T^{-1}),
\]
as $T \to \infty$.
In summary, we have
\begin{align*}
\sqrt{T b_T} L_1 &\to 0,\\
\sqrt{T b_T} L_2 &\to 0,
\end{align*}
since $b_T = o(T^{-1/5})$.

Finally, we apply Corollary \ref{cor:b11} to $L_3$ to show
\begin{equation} \label{eq:rr8.1}
\sqrt{T b_T}
\bigl(
\mathcal{L}_T(\btheta, u) - 
\mathcal{L}(\btheta, u)
\bigr)
\dlim
\mathcal{N}(0, \mathbb{V}^{\mathcal{L}}(u)).
\end{equation}
In fact, we only have to check Assumptions \ref{asp:b1} and \ref{asp:b2} for
\begin{equation} \label{eq:rr8,2}
\bm{\phi}(u, \lambda) = K(u) \spf{-1},
\end{equation}
or equivalently, $\psi(u, \lambda) = K(u) f^{ij}_{\btheta}(\lambda)$ for $i, j = 1, \dots, p$, which is expressed in the Einstein notation.
From the definition \eqref{eq:2.11} of the time-varying spectral density matrix,
$\spf{-1}$ is obviously Hermitian.
Additionally, Assumption \ref{asp:b1} (ii) is satisfied 
if both $K$ and $\spf{-1}$ are bounded functions of bounded variation, 
which follows Assumptions \ref{asp:3.1} (ii) and \ref{asp:3.2} (iii). 
Applying Corollary \ref{cor:b11} to \eqref{eq:rr8,2}, we obtain \eqref{eq:rr8.1}. 
\end{proof}

\subsection{Proof of Theorem \ref{thm:3.3}}
\begin{proof}
Note that we have
\[
\sqrt{T b_T}
\bigl(
\mathcal{L}_T(\btheta, u) - 
\mathcal{L}(\btheta, u)
\bigr)
\dlim
\mathcal{N}(0, \mathbb{V}^{\mathcal{L}}(u)).
\]
The consequence \eqref{eq:3.3} follows, if the following conditions are guaranteed for the theorem, i.e.,
\begin{enumerate}[(i)]
\item both $\mathcal{L}_T(\btheta, u)$ and $\mathcal{L}(\btheta, u)$
are convex in $\btheta$ for each $u$
and continuous in $u$ for each $\btheta$;
\item $\btheta_0(u)$ is the unique minimizer of $\mathcal{L}(\btheta, u)$
for each $u \in [0, 1]$
\end{enumerate}

According to (i), the convexity of $\mathcal{L}_T(\btheta, u)$ and $\mathcal{L}(\btheta, u)$ in $\btheta$
follows from Assumption \ref{asp:3.2} (v-b).
Especially, note that $\mathcal{L}_T(\btheta, u)$
is a linear combination of 
$
\freqint 
\log \det \pf +
\tr \Bigl(\bI_{T}(u_k, \lambda) 
\spf{-1}
\Bigr)
d\lambda$ with nonnegative coefficients, which implies that $\mathcal{L}_T(\btheta, u)$ is convex.
The continuity of $\mathcal{L}_T(\btheta, u)$ and $\mathcal{L}(\btheta, u)$ in $u$ follows from
Assumption \ref{asp:3.1}, i.e., the continuity of $K$ and $\bm{f}(\cdot, \lambda)$.
According to (ii), it is assumed in Assumption \ref{asp:3.2} (v-a).
Since $\btheta_0(u)$ is the unique minimizer and $\pf$ is twice continuously differentiable
with respect to $\btheta$, again, \eqref{eq:3.4} follows from Corollary \ref{cor:b11}. 
\end{proof}

\subsection{Proof of Theorem \ref{thm:3.7}}
\begin{proof}
For simplicity, denote 
\[
\bm{f}^Z_{\btheta(u)}(\lambda)_{11}
:=
\bm{f}_{\btheta(u)}(\lambda)_{11}- 2 \pi 
\bm{\frak{f}}_{\btheta(u)}(\lambda)_{12}
\left(\tilde{\Sigma}_{\btheta(u), 22} \right)^{-1} 
\bm{\frak{f}}_{\btheta(u)}(\lambda)_{21}. 
\]
Accordingly, 
${\rm GC}^{(2 \to 1)}(u; \btheta)$ in \eqref{eq:3.11} is simply
\[
{\rm GC}^{2 \to 1}(u; \btheta) = 
\frac{1}{2 \pi} \int_{-\pi}^{\pi} 
\log \frac{\abs{\bm{f}_{\btheta(u)}(\lambda)_{11}}}
{\abs{\bm{f}^{Z}_{\btheta(u)}(\lambda)_{11}}}
\dif \lambda.
\]
Note that the domain of the integration is bounded, and under Assumption \ref{asp:3.3},
$\log \abs{\pfu_{11}}$ is integrable in $\lambda$ for $u \in [0, 1]$, which implies that 
$\log \abs{\bm{f}^{Z}_{\btheta(u)}(\lambda)_{11}}$ is also integrable.

Now if we show $\bm{f}_{\btheta(u)}(\lambda)_{11}$ and $\bm{f}^{Z}_{\btheta(u)}(\lambda)_{11}$
are continuously differentiable with respect to $\btheta$, then
applying the delta-method to \eqref{eq:3.3} leads to the conclusion.
We summarize the parametric expressions used in the causality measure.
Suppose that $\bm{\frak{f}}_{\bm{\theta}(u)}(\lambda)$ admits the decomposition
\begin{align*}
\bm{\frak{f}}_{\bm{\theta}(u)}(\lambda) = 
\begin{bmatrix}
\bm{\frak{f}}_{\bm{\theta}(u)}(\lambda)_{11} & \bm{\frak{f}}_{\bm{\theta}(u)}(\lambda)_{12}  \\
\bm{\frak{f}}_{\bm{\theta}(u)}(\lambda)_{21} & \bm{\frak{f}}_{\bm{\theta}(u)}(\lambda)_{22} 
\end{bmatrix}.
\end{align*}
With an abuse of notation, under Assumption \ref{asp:3.3}, $\bm{f}_{\btheta(u)}$, 
defined on the unit disk $\mathcal{D}$ in the complex plane, can be factorized as
\begin{equation} \label{eq:3.6}
\bm{f}_{\bm{\theta}(u)}(z) 
= 
\frac{1}{2 \pi} \bm{\Lambda}_{\bm{\theta}(u)}(z) \bm{\Lambda}_{\bm{\theta}(u)}(z)^*,
\quad
z \in \mathcal{D}.
\end{equation}
Especially, as shown in \cite{rozanov1967}, it holds that  
\begin{equation} \label{eq:3.7}
\Sigma_{\bm{\theta}(u)} 
=
\bm{\Lambda}_{\bm{\theta}(u)}(0) \bm{\Lambda}_{\bm{\theta}(u)}(0)^*.
\end{equation}
From Lemmas 2.2 and 2.3 in \cite{hosoya1991}, we have
\begin{align}
\bm{\frak{f}}_{\bm{\theta}(u)}(\lambda)_{11} &= \bm{f}_{\bm{\theta}(u)}(\lambda)_{11}, \quad \quad
\bm{\frak{f}}_{\bm{\theta}(u)}(\lambda)_{12} = \bm{\frak{f}}_{\bm{\theta}(u)}(\lambda)_{21}^*, 
\label{eq:3.8}\\
\bm{\frak{f}}_{\bm{\theta}(u)}(\lambda)_{21} &= 
\begin{bmatrix}
- \Sigma_{\bm{\theta}(u), 21} \Sigma_{\bm{\theta}(u), 11}^{-1} & I_{M}
\end{bmatrix}
\bm{\Lambda}_{\bm{\theta}(u)}(0) \bm{\Lambda}_{\bm{\theta}(u)}(e^{\I \lambda})^{-1} 
\begin{pmatrix}
\bm{f}_{\bm{\theta}(u)}(\lambda)_{11} \\
\bm{f}_{\bm{\theta}(u)}(\lambda)_{21} 
\end{pmatrix}, \label{eq:3.9}
\end{align}
and 
\begin{equation} \label{eq:3.10}
\bm{\frak{f}}_{\bm{\theta}(u)}(\lambda)_{22} 
= \frac{1}{2 \pi} \tilde{\Sigma}_{\bm{\theta}(u), 22}
:= \frac{1}{2 \pi}  \left\{ \Sigma_{\bm{\theta}(u), 22} - \Sigma_{\bm{\theta}(u), 21} 
\Sigma_{\bm{\theta}(u), 11}^{-1} \Sigma_{\bm{\theta}(u), 12} \right\}.
\end{equation}

The continuous differentiability of $\pfu_{11}$ with respect to $\btheta$
directly follows from that of $\pfu$ under Assumption \ref{asp:3.2} (iv).
Note that $\tilde{\Sigma}_{\btheta(u), 22}$ is a continuous function of $\Sigma_{\btheta(u)}$ from \eqref{eq:3.10}.
Using the expression for $\Sigma_{\btheta(u)}$ in \eqref{eq:3.7}
and the relation in \eqref{eq:3.6}, the continuous differentiability of $\Sigma_{\btheta(u)}$ with respect to $\btheta$
follows from that of $\pfu$.
In addition, this implies the continuous differentiability of $\bm{\frak{f}}_{\btheta(u)}(\lambda)_{21}$ from \eqref{eq:3.9},
which in turn implies the continuous differentiability of $\bm{f}^{Z}_{\btheta(u)}(\lambda)_{11}$.
This completes the proof of Theorem \ref{thm:3.7}.
\end{proof}

\subsection{Proof of Theorem \ref{thm:3.8}}
\begin{proof}
By Theorem 6.8 of \cite{mn2007}, we have
\begin{multline*}
{\rm GC}^{(2 \to 1)}(u;\,
{\hat{\bm{\theta}}_T}) - 
{\rm GC}^{(2 \to 1)}(u; {\bm{\theta}_0})
=
\nabla {\rm GC}^{(2 \to 1)}(u; {\bm{\theta}_0})^{\top} 
\bigl(\hat{\btheta}_T(u) - \btheta_0(u)\bigr)\\
+
\frac{1}{2}
\bigl(\hat{\btheta}_T(u) - \btheta_0(u)\bigr)^{\top} 
\mathcal{H}(u) 
\bigl(\hat{\btheta}_T(u) - \btheta_0(u)\bigr) + o_P\Bigl(\bigl(\hat{\btheta}_T(u) - \btheta_0(u)\bigr) ^2\Bigr).
\end{multline*}
Since $\nabla {\rm GC}^{(2 \to 1)}(u; {\bm{\theta}_0}) = \bzero$ from \eqref{eq:3.23},
we have
\begin{multline*}
Tb_T
\bigl(
{\rm GC}^{(2 \to 1)}(u;\,
{\hat{\bm{\theta}}_T}) - 
{\rm GC}^{(2 \to 1)}(u; {\bm{\theta}_0})
\bigr)
=\\
\frac{1}{2}
\sqrt{Tb_T}\bigl(\hat{\btheta}_T(u) - \btheta_0(u)\bigr)^{\top} 
\mathcal{H}(u) 
\sqrt{Tb_T}\bigl(\hat{\btheta}_T(u) - \btheta_0(u)\bigr) + o_P(1).
\end{multline*}
We arrived at the conclusion \eqref{eq:3.25} by the continuous mapping theorem.
\end{proof}


\section{Multivariate locally stationary processes} \label{app:a}
In Section 
\ref{app:a}, we review the basic properties of multivariate locally stationary processes.
Especially, we evaluate the absolute differences of the covariances and higher order cumulants 
between a multivariate locally stationary process
$\{\bx{t}\}$ and the approximate stationary process $\{\bm{X}^*(u)\}$ with a spectral density matrix $\bm{f}(u, \lambda)$.\\

Let $l(j)$ be
\[
l(j): = 
\begin{cases}
1, &\qquad \abs{j} \leq 1, \\
\abs{j} \log^{1 + \kappa}\abs{j}, &\qquad \abs{j} > 1,
\end{cases} 
\]
for some constant $\kappa > 0$.
Let $C$ be a generic constant in Appendix, and
the following inequality is repetitively used in the proof.
\begin{equation} \label{eq:a1}
\tsum[j]{\infty} \frac{1}{l(j) l(j + s)}
\leq
\frac{C}{l(s)}.
\end{equation}

The following assumption corresponds to Assumption \ref{asp:2.1}, which is imposed for the multivariate locally stationary process $\{\bx{t}\}$ in the main text.
\begin{asp} \label{asp:a1}
Suppose the multivariate locally stationary process 
$\bx{t} = (X_{t, T}^{(1)}, \dots, \\X_{t, T}^{(d)}, \dots, X_{t, T}^{(p)})^{\top}$ has a representation
\begin{equation} \label{eq:a2}
\bx{t} = \tsum[j]{\infty} A_{t, T}(j) \be_{t - j},
\end{equation}
where the sequences $\{A_{t, T}(j)\}_{j \in \Z}$ and $\{\be_t\}_{t \in \Z}$ satisfy the following conditions:
there exists a constant $C_A$ such that
\begin{equation} \label{eq:r56}
\sup_{t, T} \norm{A_{t, T}(j)}_{\infty} \leq \frac{C_A}{l(j)},
\end{equation}
and there exists a sequence of functions $A(\cdot, j): [0, 1] \to \R$ such that
\begin{enumerate}[(i)]
\item $\sup_{u} \norm{A(u, j)}_{\infty} \leq \frac{C_A}{l(j)}$;
\item $\sup_j \sumn[t] \Bnorm{A_{t, T}(j) - A\Bigl(\frac{t}{T}, j\Bigr)}_{\infty} \leq C_A$;
\item $V\Bigl(\norm{A(\cdot, j)}_{\infty}\Bigr) \leq \frac{C_A}{l(j)}$,
\end{enumerate}
where $V(f)$ is the total variation of the function $f$ on the interval $[0, 1]$, i.e., $V$  is defined as
\[
V(f) = \sup\Bigl\{
\sum_{k = 1}^m \abs{f(x_k) - f(x_{k - 1})}; \,
0 \leq x_0 < \cdots < x_m \leq 1, \, m \in \N
\Bigr\}.
\]
In addition, 
the $\be_t$ are assumed to be independent and identically distributed with
$E \be_t = \bzero$ and $E \be_t \be_t^{\top} = K$,
where the matrix $K$ exists and all elements are bounded by $C_K$.
Furthermore, all the moments of $\be_t$ exist.
All elements in $r$th moment of $\be_t$ are bounded by $C_{\bm{\ep}}^{(r)} < C$ for each $r \geq 3$ and some finite constant $C > 0$. 
\end{asp}

\begin{rem} \label{rem:a1}
Let us consider the time-varying spectral density matrix 
$\fu = \bigl(\fu_{ij}\bigr)_{i, j = 1, \dots, p}$ defined in \eqref{eq:2.11}, i.e.,
\begin{equation*} 
\fu = \frac{1}{2\pi} A(u, \lambda) K A(u, -\lambda)^{\top},
\end{equation*}
where $A(u, \lambda) = \sum_{j = -\infty}^{\infty} A(u, j) \exp(\I j \lambda)$.
The autovariance function $\bm{\gamma}(u, s)$ at $u$ is 
\begin{equation} \label{eq:a3}
\bm{\gamma}(u, s)
=
\freqint \fu \exp(\I \lambda s) \dif \lambda
=
\sum_{j=-\infty}^{\infty} A(u, j) K A(u, j + s)^{\top}.
\end{equation}
Especially, the (a, b)-element of the matrix $\bm{\gamma}$ is bounded by
\begin{eqnarray}
\abs{\bm{\gamma}(u, s)_{ab}}
&\leq&
\sum_{j=-\infty}^{\infty} 
\sup_u \norm{A(u, j)}_{\infty} \norm{K}_{\infty}  \sup_{u}\norm{A(u, j + s)^{\top}}_{\infty} \notag\\
&\leq&
C \sum_{j=-\infty}^{\infty} \frac{1}{l(j) l(j + s)} \notag\\
&\leq&
\frac{C}{l(s)}, \label{eq:a4}
\end{eqnarray}
where the second inequality follows from Assumption \ref{asp:a1} (i) and the third inequality follows from \eqref{eq:a1}.
\end{rem}

\begin{rem}
From \eqref{eq:a4}, we can see that 
for any fixed $u \in [0, 1]$ and any $(a, b)$-element of the autovariance matrix,
$\abs{\bm{\gamma}(u, 0)_{ab}}$ is bounded, i.e.,
\[
\abs{\bm{\gamma}(u, 0)_{ab}} \leq C.
\]
Thus, the time-varying spectral density $\fu_{jk}$ ($1 \leq j, k \leq p$)
are square-integrable for any fixed $u \in [0, 1]$.
\end{rem}

\begin{rem}
Under Assumption \ref{asp:a1}, the locally stationary process $\{\bx{t}\}$
has the following properties.
Let $\bx{t}^{(d)}$ be the $d$th element of the vector $\bx{t}$.
From \eqref{eq:a2}, $\bx{t}^{(d)}$ has the expression
\[
\bx{t}^{(d)}
=
\tsum[j]{\infty}\sum_{m = 1}^p A_{t, T}(j)_{dm} \be_{t - j}^{(m)}.
\]
Thus, we obtain
\begin{align}
\cov(\bx{t}^{(a)}, \bx{t + s}^{(b)})
&=
\sum_{j = -\infty}^{\infty} \sum_{l = -\infty}^{\infty}
\sum_{m, n = 1}^p
A_{t, T}(j)_{am}
A_{t + s, T}(l)_{bn}
\cov\Bigl(
\be_{t - j}^{(m)},
\be_{t + s - l}^{(n)}
\Bigr) \notag \\
&=
\sum_{j = -\infty}^{\infty}
\sum_{m, n = 1}^p
A_{t, T}(j)_{am}
K_{mn}
A_{t + s, T}(j + s)_{bn} \label{eq:a5} \\
&=
\tsum[j]{\infty} \Bigl(A_{t, T}(j) K A_{t + s, T}(j)^{\top}\Bigr)_{ab}. \label{eq:a6}
\end{align}
\end{rem}

We first clarify the difference between \eqref{eq:a3} and \eqref{eq:a6}
on discrete points $u_k = k/T$ in the following.
\begin{lem} \label{lem:a3}
Under Assumption \ref{asp:a1}, we have
\begin{equation} \label{eq:a7}
\sum_{k =1}^T
\Babs{
\cov\Bigl(
\bx{[k+ 1/2 - s/2]}^{(a)}, \bx{[k + 1/2 + s/2]}^{(b)}
\Bigr) 
- \bm{\gamma}(u_k, s)_{ab}
}
\leq
C\Bigl(
1 + \frac{1}{l(s)}
\Bigr),
\end{equation}
where $C$ is a generic constant.
\end{lem}

\begin{proof}
To evaluate \eqref{eq:a7}, we use the expressions \eqref{eq:a5} and \eqref{eq:a3}.
Note that
\begin{multline}
\sum_{k =1}^T
\Babs{
\cov\Bigl(
\bx{[k+ 1/2 - s/2]}^{(a)}, \bx{[k + 1/2 + s/2]}^{(b)}
\Bigr) 
- \bm{\gamma}(u_k, s)_{ab}
} \\
\leq
\sum_{k = 1}^T
\Babs{
\sum_{j = -\infty}^{\infty}
\sum_{m, n = 1}^p
\Bigl(
A_{[k+ 1/2 - s/2], T}(j)_{am}
K_{mn}
A_{[k+ 1/2 + s/2], T}(j + s)_{bn}\\
- 
A(u_k, j)_{am}
K_{mn}
A_{[k+ 1/2 + s/2], T}(j + s)_{bn}
\Bigr)
} \\
+
\sum_{k = 1}^T
\Babs{
\sum_{j = -\infty}^{\infty}
\sum_{m, n = 1}^p
\Bigl(
A(u_k, j)_{am}
K_{mn}
A_{[k+ 1/2 + s/2], T}(j + s)_{bn}\\
-
A(u_k, j)_{am}
K_{mn}
A(u_k, j + s)_{bn}
\Bigr)
}. \label{eq:r62}
\end{multline}
Considering the first term in the right hand side, we have
\begin{align}
&
\sum_{k = 1}^T
\Babs{
\sum_{j = -\infty}^{\infty}
\sum_{m, n = 1}^p
\Bigl(
A_{[k+ 1/2 - s/2], T}(j)_{am}
K_{mn}
A_{[k+ 1/2 + s/2], T}(j + s)_{bn} \notag \\
& \qquad   \qquad  \qquad  \qquad  \qquad- 
A(u_k, j)_{am}
K_{mn}
A_{[k+ 1/2 + s/2], T}(j + s)_{bn}
\Bigr)
} \notag \\
= &
\sum_{j = -\infty}^{\infty}
\sum_{k = 1}^T
\Babs{
\sum_{m, n = 1}^p
\Bigl(
A_{[k+ 1/2 - s/2], T}(j) -
A(u_k, j)
\Bigr)_{am}
K_{mn}
A_{[k+ 1/2 + s/2], T}(j + s)_{bn}
} \notag \\
\leq &
\sum_{j = -\infty}^{\infty}
\sum_{k = 1}^T
\sum_{m, n = 1}^p
\Babs{
\Bigl(
A_{[k+ 1/2 - s/2], T}(j) -
A(u_k, j)
\Bigr)_{am}}
\abs{
K_{mn}
}\,\,
\Babs{
A_{[k+ 1/2 + s/2], T}(j + s)_{bn}
} \notag \\
\leq &
\sum_{j = -\infty}^{\infty}
\frac{C_K C_A^2}{l(j + s)}
+
\frac{C_K C_A^2}{l(j) l(j + s)}
\notag\\
\leq &
C C_K C_A^2, \label{eq:a8}
\end{align}
where the first inequality follows from $\abs{K_{mn}} \leq C_K$,
$\Babs{
A_{k+ 1/2 + s/2, T}(j + s)_{bn}
} \leq C_A/l(j + s)$ 
from \eqref{eq:r56},
and
\begin{align*}
&\sum_{k = 1}^T
\sum_{m, n = 1}^p
\Babs{
\Bigl(
A_{[k+ 1/2 - s/2], T}(j) -
A(u_k, j)
\Bigr)_{am}} \\
\leq &
\sum_{k = 1}^T
\sum_{m, n = 1}^p
\Babs{
\Bigl(
A_{[k+ 1/2 - s/2], T}(j) -
A([k+ 1/2 - s/2]/T, j)
\Bigr)_{am}} \\
& \qquad \qquad \qquad \qquad \qquad +
\Babs{
\Bigl(
A([k+ 1/2 - s/2]/T, j) -
A(u_k, j)
\Bigr)_{am}} \\
\leq &
C_A + \frac{C_A}{l(j)},
\end{align*}
where the second inequality follows from
(ii) and (iii) in Assumption \ref{asp:a1}.

Also, it holds that
\begin{align}
&
\sum_{k = 1}^T
\Babs{
\sum_{j = -\infty}^{\infty}
\sum_{m, n = 1}^p
\Bigl(
A(u_k, j)_{am}
K_{mn}
A_{[k+ 1/2 + s/2], T}(j + s)_{bn}
-
A(u_k, j)_{am}
K_{mn}
A(u_k, j + s)_{bn}
\Bigr)
} \notag \\
&=
\sum_{j = -\infty}^{\infty}
\sum_{k = 1}^T
\Babs{
\sum_{m, n = 1}^p
A(u_k, j)_{am}
K_{mn}
\Bigl(
A_{[k+ 1/2 + s/2], T}(j + s)
-
A(u_k, j + s)
\Bigr)_{bn}
} \notag \\
&\leq
\sum_{j = -\infty}^{\infty}
\sum_{k = 1}^T
\sum_{m, n = 1}^p
\Babs{
A(u_k, j)_{am}
}
\abs{
K_{mn}
}
\Babs{
\Bigl(
A_{[k+ 1/2 + s/2], T}(j + s)
-
A([k+ 1/2 + s/2]/T, j + s)
\Bigr)_{bn}
} \notag \\
&\quad
+
\sum_{j = -\infty}^{\infty}
\sum_{k = 1}^T
\sum_{m, n = 1}^p
\Babs{
A(u_k, j)_{am}
}
\abs{
K_{mn}
}
\Babs{
\Bigl(
A([k+ 1/2 + s/2]/T, j + s)
-
A(u_k, j + s)
\Bigr)_{bn}
} \notag \\
&\leq
\sum_{j = -\infty}^{\infty}
\frac{C_K C_A^2}{l(j)} + 
\sum_{j = -\infty}^{\infty}
\frac{C_K C_A}{l(j)} \sum_{k = 1}^T \norm{A([k+ 1/2 + s/2]/T, j + s) - A(u_k, j + s)}_{\infty} \notag \\
&\leq
C C_K C_A^2 + \sum_{j = -\infty}^{\infty}\frac{C_K C_A^2}{l(j) l(j + s)} \notag\\
&\leq
C C_K C_A^2\Bigl(
1 + \frac{1}{l(s)}
\Bigr), \label{eq:a9}
\end{align}
where the last inequality follows from \eqref{eq:a1}.
Combining \eqref{eq:r62}, \eqref{eq:a8} and \eqref{eq:a9}, we obtain the desired result.
\end{proof}

Generally, higher-order cumulants of the locally stationary process $\{\bx{t}\}$ 
can be approximated by those of the stationary process $\{\bm{X}^*(u)\}$ for $u = t/T$ under Assumption \ref{asp:a1}
in a similar manner as the autocovariance. To discuss higher-order cumulants,
we introduce the notation $\bm{X}(u; s)$, 
which means the observation $\bm{X}(s)$, $s \in \Z$, of the stationary process $\bm{X}^*(u)$.

Let $\gamma_{a_1, \dots, a_q}(u;\, t_1, \dots, t_{q-1})$ be the joint cumulant function of order $q$, i.e.,
\begin{multline*}
\gamma_{a_1, \dots, a_q}(u;\, t_1, \dots, t_{q-1})
:=
\cum\{
\bm{X}^{(a_1)}(u; t + t_1), 
\bm{X}^{(a_2)}(u; t + t_2),\\
\cdots,
\bm{X}^{(a_2)}(u; t + t_{q-1}),
\bm{X}^{(a_q)}(u; t)
\}.
\end{multline*}
To discriminate this notation from the autocovariance function,
we do not use $\gamma$ in boldface, although it is an extension of the autocovariance function to higher-orders.

\begin{lem} \label{lem:a4}
Under Assumption \ref{asp:a1}, we have
\begin{multline*}
\sum_{k =1}^T
\Babs{
\cum(
\bx{k + t_1}^{(a_1)}, 
\bx{k + t_2}^{(a_2)},
\cdots
\bx{k + t_{q-1}}^{(a_{q-1})},
\bx{k}^{(a_{q})}
)
- 
\gamma_{a_1, \dots, a_q}(u_k;\, t_1, \dots, t_{q-1})
}\\
\leq
C
\sum_{m = -\infty}^{\infty}
\Biggl(
\sum_{i = 1}^{q} \frac{l(m + t_i)}{\prod_{j = 1}^{q}l(m + t_j)}
+
\frac{1}{
\prod_{j = 1}^{q}l(m + t_j)}
\Biggr),
\end{multline*}
where $C$ is a generic constant and $t_q = 0$.
\end{lem}

\begin{rem}
Lemma \ref{lem:a3} is a special case of Lemma \ref{lem:a4} when $q = 2$.
\end{rem}

\begin{proof}
Under Assumption \ref{asp:a1}, there exists a constant $\tilde{C}_{\ep}^{(r)}$
such that all cumulants of order $r$ are all bounded by $\tilde{C}_{\ep}^{(r)}$,
since all cumulants can be written in the form of polynomials of moments.
Now, it holds that
\begin{eqnarray*}
&&\Babs{\cum(
\bx{k + t_1}^{(a_1)}, 
\bx{k + t_2}^{(a_2)},
\cdots,
\bx{k + t_{q-1}}^{(a_{q-1})},
\bx{k}^{(a_{q})})
- 
\gamma_{a_1, \dots, a_q}(u_k;\, t_1, \dots, t_{q-1})
}\\
&\leq&
\Babs{
\sum_{j_1, \dots, j_q = -\infty}^{\infty}
\sum_{m_1, \dots, m_q = 1}^p
\Bigl(
A_{k + t_1, T}(j_1)_{a_1 m_1} \cdots A_{k, T}(j_q)_{a_q m_q}\\
&&-
A(u_k, j_1)_{a_1 m_1} \cdots A(u_k, j_q)_{a_q m_q}
\Bigr)
\cum\Bigl(
\be_{k + t_1 - j_1}^{(m_1)},
\be_{k + t_2 - j_2}^{(m_2)},
\cdots,
\be_{k - j_q}^{(m_q)}
\Bigr)}\\
&\leq&
\tilde{C}_{\ep}^{(q)}
\sum_{j_q = -\infty}^{\infty}
\sum_{m_1, \dots, m_q = 1}^p
\Babs{
A_{k + t_1, T}(j_q + t_1)_{a_1 m_1} \cdots A_{k, T}(j_q)_{a_q m_q}\\
&&-
A(u_k, j_q + t_1)_{a_1 m_1} \cdots A(u_k, j_q)_{a_q m_q}
}\\
&\leq&
\tilde{C}_{\ep}^{(q)}  \sum_{j_q = -\infty}^{\infty}
\norm{A_{k + t_1, T}(j_q + t_1)_{a_1 m_1} \cdots A_{k, T}(j_q)_{a_q m_q} \\
&&
-
A(u_k, j_q + t_1)_{a_1 m_1} A_{k + t_2, T}(j_q + t_2)_{a_q m_q} \cdots A_{k, T}(j_q)_{a_q m_q}
}_{\infty} \\
&&+
\norm{
A(u_k, j_q + t_1)_{a_1 m_1} A_{k + t_2, T}(j_q + t_2)_{a_2 m_2} \cdots A_{k, T}(j_q)_{a_q m_q} \\
&&
-
A(u_k, j_q + t_1)_{a_1 m_1} A(u_k, j_q + t_2)_{a_2 m_2} \cdots A_{k, T}(j_q)_{a_q m_q}
}_{\infty} \\
&&+
\norm{A(u_k, j_q + t_1)_{a_1 m_1} \cdots A_{k, T}(j_q)_{a_q m_q} 
-
A(u_k, j_q + t_1)_{a_1 m_1} \cdots A(u_k, j_q)_{a_q m_q}
}_{\infty} 
\end{eqnarray*}
Note that, for $1 \leq i \leq q - 1$, we have
\begin{eqnarray*}
&& \sum_{k = 1}^T
\norm{
A_{k + t_i, T}(j_q + t_i)
-
A(u_k, j_q + t_i)
}_{\infty}\\
&\leq&
\sum_{k = 1}^T
\Bigl(
\norm{
A_{k + t_i, T}(j_q + t_i)
-
A(u_{k + t_i}, j_q + t_i)
}_{\infty}
+
\norm{
A(u_{k + t_i}, j_q + t_i)
-
A(u_k, j_q + t_i)
}_{\infty}
\Bigr)
\\
&\leq&
C_A + \frac{C_A}{l(j_q + t_i)}.
\end{eqnarray*}
Thus, it holds 
\begin{eqnarray*}
&&\sum_{k =1}^T
\Babs{
\cum(
\bx{k + t_1}^{(a_1)}, 
\bx{k + t_2}^{(a_2)},
\cdots
\bx{k + t_{K-1}}^{(a_{q-1})},
\bx{k}^{(a_{q})}
)
- 
\gamma_{a_1, \dots, a_1}(u_k;\, t_1, \dots, t_{q-1})
}\\
&\leq&
\tilde{C}_{\ep}^{(q)} C_A^{q}
\sum_{j_q = -\infty}^{\infty}
\Biggl(
\sum_{i = 1}^{q - 1} \frac{l(j_q + t_i)}{l(j_q) \prod_{j = 1}^{q - 1}l(j_q + t_j)}
+
\frac{1}{
l(j_q) \prod_{j = 1}^{q - 1}l(j_q + t_j)}
+
\frac{1}
{\prod_{j = 1}^{q - 1}l(j_q + t_j)}
\Biggr)
\end{eqnarray*}
We obtain the conclusion if we replace $j_q$ with $m$ and set $t_q = 0$.
\end{proof}

\section{Empirical spectral process} \label{app:b}
In Section 
\ref{app:b}, we consider the asymptotic distribution of the empirical spectral process for multivariate locally stationary processes.

We first impose the following assumptions on the matrix-valued functions $\bp$, which is to be considered later.
Let $V_2(\cdot)$ be the total variation of bivariate functions, i.e.,
\begin{multline*}
V_2(f) =
\sup\Bigl\{
\sum_{k, l = 1}^{m, n} \abs{f(u_k, \lambda_l) - f(u_{k - 1}, \lambda_l)
- f(u_k, \lambda_{l - 1}) + f(u_{k - 1}, \lambda_{l - 1})
};\, \\
0 \leq u_0 < \cdots < u_m \leq 1, \,
0 \leq \lambda_0 < \cdots < \lambda_n \leq \pi; \,
\, m, n \in \N
\Bigr\}.
\end{multline*}

Let $\Psi$ be a class of square-integrable functions, where 
the $\mathcal{L}_2$-norm on $\psi \in \Psi$ is defined as
\[
\norm{\psi}_{\mathcal{L}_2}^2 = \uint \freqint \psi(u, \lambda)^2 \dif \lambda \dif u < \infty.
\]
The class $\Psi$ is considered for the elements of the matrix $\bp$.

For any class $\Phi: = \{\bp \in \R^{p \times p}; \phi_{ij} \in \Psi \text{ for $i, j = 1, \dots, p$}\}$, 
let $\tau_{\infty, \rm{TV}}$, $\tau_{\rm{TV}, \infty}$, $\tau_{\rm{TV}, \rm{TV}}$ and 
$\tau_{\infty, \infty}$ be
\begin{align*}
\tau_{\infty, \rm{TV}} &:= \tau_{\infty, \rm{TV}}(\Phi) 
= \sup_{\phi \in \Phi}\max_{1 \leq i, j \leq p} \sup_{u \in [0, 1]} V\bigl(\phi_{ij}(u, \cdot)\bigr), \\
\tau_{\rm{TV}, \infty} &:= \tau_{\rm{TV}, \infty}(\Phi) 
= \sup_{\phi \in \Phi}\max_{1 \leq i, j \leq p} \sup_{\lambda \in [0, \pi]} V\bigl(\phi_{ij}(\cdot, \lambda)\bigr), \\
\tau_{\rm{TV}, \rm{TV}} &:= \tau_{\rm{TV}, \rm{TV}}(\Phi) 
=  \sup_{\phi \in \Phi}\max_{1 \leq i, j \leq p} V_2\bigl(\phi_{ij}\bigr), \\
\tau_{\infty, \infty} &:= \tau_{\infty, \infty}(\Phi) =  \sup_{\phi \in \Phi}\max_{1 \leq i, j \leq p} \sup_{\substack{u \in [0, 1]\\\lambda \in [0, \pi]}} \abs{\phi_{ij}}.
\end{align*}

\begin{asp} \label{asp:b1}
Let $\Phi$ be a class of $p \times p$ matrix-valued continuous functions $\bp(u, \cdot)$ on $[-\pi, \pi]$
such that for any $\bp \in \Phi$, it holds that
\begin{enumerate}[(i)]
\item $\bp(u, \cdot) = \bp^*(u, \cdot)$ for any fixed $u \in [0, 1]$;
\item $\tau_{\infty, \rm{TV}}$, $\tau_{\rm{TV}, \infty}$, $\tau_{\rm{TV}, \rm{TV}}$ and 
$\tau_{\infty, \infty}$ are all finite.
\end{enumerate}
\end{asp}

For any function $\psi \in \Psi$, let $\psi_T$ be
\begin{equation} \label{eq:b1}
\psi_T(u, \lambda) = \frac{1}{b_T} \psi \Bigl(
\frac{u}{b_T}, \lambda
\Bigr),
\end{equation}
where $b: = b_T \to 0$ as $T \to \infty$.
Let $\Psi_T$ denotes the function class constituted by $\psi_T$, i.e.,
\begin{equation} \label{eq:function_class}
\Psi_T = \{
\psi_T(u, \lambda) = \frac{1}{b_T} \psi \Bigl(
\frac{u}{b_T}, \lambda
\Bigr);
\,
\psi \in \mathcal{L}_2
\}.
\end{equation}

\begin{asp} \label{asp:b2}
For any $\psi \in \Psi$,
let $\psi(\cdot, \lambda)$ be a positive, symmetric function of bounded variation such that
$\psi(\cdot, \lambda)$ has a compact support on $[-1, 1]$.
\end{asp}

Let $\mathscr{A}_T(u)$ and $\bar{\mathscr{A}}_T(u)$ be
\begin{align}
\mathscr{A}_T(u)_{ab}:= 
\mathscr{A}_T(u; \psi)_{ab} 
&= 
\frac{1}{T} \sumn[k]  \freqint \psi_T(u - u_k, \lambda) \bI_T(u_k, \lambda)_{ab} \dif\lambda, \label{eq:b3} \\
\bar{\mathscr{A}}_T(u)_{ab}:= 
\bar{\mathscr{A}}_T(u; \psi)_{ab} &= 
\frac{1}{T}  \sumn[k] \freqint \psi_T(u - u_k, \lambda) \fu[u_k]_{ab} \dif \lambda. \label{eq:b4} 
\end{align}
The empirical spectral process $\xi_T(u)_{ab}$ is
\begin{equation} \label{eq:b5} 
\xi_T(u)_{ab}
:=
\xi_T(u; \psi)_{ab}
=
\sqrt{Tb_T} 
\bigl(\mathscr{A}_T(u; \psi) - \bar{\mathscr{A}}_T(u; \psi)\bigr)_{ab}.
\end{equation}
We use the first expression in \eqref{eq:b3}, \eqref{eq:b4} and \eqref{eq:b5}
when there is no confusion with $\psi$.

\subsection{Preliminary Computations}
Let $\hat{\psi}$ be
\[
\hat{\psi}(u, k) = 
\freqint
\psi(u, \lambda) \exp(-i k \lambda) \dif \lambda.
\]

\begin{lem}
Let $\beta_T$ be a sequence of positive numbers such that $\beta_T \to 0$ as $T \to \infty$.
Suppose
\begin{equation} \label{eq:b6}
\limsup_{T \to \infty} \,\,  \beta_T
\tsum[s]{T} \sup_u \abs{\hat{\psi}(u, -s)} < \infty.
\end{equation}
Then, it holds that
\[
\abs{E \mathscr{A}_T(u)_{ab} - \bar{\mathscr{A}}_T(u)_{ab}} = O(T^{-1} b_T^{-1} \beta_T^{-1}).
\]
\end{lem}

\begin{rem}
The condition
\begin{equation} \label{eq:b7}
\tsum[s]{\infty} \sup_u \abs{\hat{\psi}(u, -s)} < \infty
\end{equation}
satisfies \eqref{eq:b6}.
However, if $\psi(u, \cdot)$ is only a function of bounded variation,
then $\psi$ may not satisfy the condition \eqref{eq:b7}.
Under \eqref{eq:b6}, we see that
\[
\tsum[s]{T} \sup_u \abs{\hat{\psi}(u, -s)} = O(\beta_T^{-1}),
\]
which we use in the following evaluations.
\end{rem}

\begin{proof}
From \eqref{eq:r3.17}, we have
\begin{equation} \label{eq:rb8}
\bI_T(u, \lambda)_{ab}
=
\frac{1}{2\pi} \sum_{l: 1 \leq [u T + 1/2 \pm l/2] \leq T}
\bx{[u T + 1/2 + l/2]}^{(a)}
{\bx{[u T + 1/2 - l/2]}^{(b)}}
\exp(- \I \lambda l).
\end{equation}
In expression \eqref{eq:rb8}, $l$ depends on $u$,
but it can be naturally extended to 
\begin{equation} \label{eq:rb9}
\bI_T(u, \lambda)_{ab}
=
\frac{1}{2\pi} \sum_{l = 1- T}^{T - 1}
\bx{[u T + 1/2 + l/2]}^{(a)}
{\bx{[u T + 1/2 - l/2]}^{(b)}}
\exp(- \I \lambda l),
\end{equation}
if we let $\bx{m} \equiv 0$ for any $m \leq 0$ or $m \geq T + 1$.
We shall use this expression \eqref{eq:rb9} in the following proof.
By Parseval's identity, it holds that
\begin{eqnarray*}
&&
\abs{E\mathscr{A}_T(u)_{ab} - \bar{\mathscr{A}}_T(u)_{ab}} \\
&\leq&
\Babs{\frac{1}{T} \sum_{k = 1}^{T}
\freqint \psi_T(u - u_k, \lambda) 
\bigl(
E\bI_T(u_k, \lambda)_{ab} - \fu[u_k]_{ab}
\bigr)
\dif\lambda}
\\
&=&
\Babs{
\frac{1}{2\pi T}
 \sum_{k = 1}^{T}
 \Bigl(
\sum_{s = 1 - T}^{T - 1}
\hat{\psi}_T(u - u_k, -s)
\Bigl(
E\bx{[k + 1/2 + s/2]}^{(a)} \bx{[k + 1/2 - s/2]}^{(b)}
-
\bm{\gamma}(u_k, -s)_{ab}
\Bigr)
\\
&&\qquad\qquad
+
\sum_{\abs{s} \geq T}
\hat{\psi}_T(u - u_k, -s)
\bm{\gamma}(u_k, -s)_{ab}\Bigr)
}\\
&\leq&
\frac{1}{2\pi T}
\Babs{
\sum_{k = 1}^{T}
\sum_{s = 1 - T}^{T - 1}
\hat{\psi}_T(u - u_k, -s)
\Bigl(
\cov(\bx{[k + 1/2 + s/2]}^{(a)}, \bx{[k + 1/2 - s/2]}^{(b)})
-
\bm{\gamma}(u_k, -s)_{ab}
\Bigr)
}
\\
&&\qquad\qquad
+ 
\frac{1}{2\pi T}
\sum_{k = 1}^{T}
\Babs{
\sum_{\abs{s} \geq T}
\hat{\psi}_T(u - u_k, -s)
\bm{\gamma}(u_k, -s)_{ab}
}\\
&:=&
B_1 + B_2, \quad \text{(say).}
\end{eqnarray*}
By Lemma \ref{lem:a3}, it holds that
\begin{eqnarray*}
B_1 
&\leq&
\frac{1}{2\pi b_T}
\sum_{s = 1 - T}^{T - 1}
\sup_u
\abs{\hat{\psi}(u, -s)}
\Babs{
\frac{1}{T}
\sum_{k = 1}^{T}
\cov(\bx{[k + 1/2 + s/2]}^{(a)}, \bx{[k + 1/2 - s/2]}^{(b)})
-
\bm{\gamma}(u_k, -s)_{ab}
}\\
&\leq&
\frac{C}{2\pi b_T T}
\sum_{s = 1 - T}^{T - 1}
\sup_u
\abs{\hat{\psi}(u, -s)}
\Bigl(
1 + \frac{1}{l(s)}
\Bigr)\\
&=&
O(T^{-1} b_T^{-1} \beta_T^{-1}).
\end{eqnarray*}
Further, noting \eqref{eq:a4}, we have
\begin{eqnarray*}
B_2 
&\leq&
\frac{1}{2\pi b_T} \sup_u
\sum_{\abs{s} \geq T}
\abs{
\hat{\psi}(u, -s)
}
\abs{
\bm{\gamma}(u, -s)_{ab}
}
\\
&\leq&
\frac{C}{b_T \beta_T}\sum_{\abs{s} \geq T}\frac{C}{l(s)} \\
&=&
O(T^{-1} b_T^{-1} \beta_T^{-1}),
\end{eqnarray*}
since $\{l(j)^{-1}\}_{j \in \N}$ is a convergent series, 
and $\abs{l(s)} \geq T$ for $\abs{s} \geq T$.
Therefore, we obtain the assertion.
\end{proof}

Next, we evaluate the higher-order cumulants of $\xi_T(u)$.
We first clarify the bias between those of the time-varying process $\{\bx{t}\}$ 
and those of the approximate stationary process $\{\bm{X}(u, t)\}$,
and then evaluate the higher order cumulants of the stationary process.

\begin{lem} \label{lem:b5}
Let $\beta_T$ be a sequence of positive numbers such that $\beta_T \to 0$ as $T \to \infty$.
Suppose $\psi^{(1)}(\cdot, \lambda)$, \dots, $\psi^{(q)}(\cdot, \lambda)$ 
are all functions of bounded variation and satisfy Assumption \ref{asp:b2} and 
\begin{equation} \label{eq:b9}
\limsup_{T\to \infty} \,\, \beta_T \tsum[s]{T} \sup_u \abs{\hat{\psi}^{(i)}(u, -s)} < \infty,
\qquad
\text{for $i = 1, \dots, q$}.
\end{equation}
If $b_T \to 0$, $T b_T \to \infty$ and $T^{-q/2} \beta_T^{-1} \to 0$ as $T \to \infty$, then it holds that
\[
\cum\bigl(
\xi_T(u^{(1)}; \psi^{(1)})_{a_1b_1},
\cdots,
\xi_T(u^{(q)}; \psi^{(q)})_{a_qb_q}
\bigr)= O(T^{1 - q/2} b_T^{1 - q/2}).
\]
Especially, when $q = 2$, we have
\begin{multline} \label{eq:b10}
\lim_{T \to \infty}
\cov\bigl(
\xi_T(u^{(1)}; \psi^{(1)})_{a_1b_1},
\xi_T(u^{(2)}; \psi^{(2)})_{a_2b_2}
\bigr)
=\\
2\pi \delta(u^{(1)}, u^{(2)}) \Biggl(
\freqint \Bigl(\int_{-\infty}^{\infty} 
\psi^{(1)}(v, \lambda) 
\overline{\psi^{(2)}(v, \lambda)} \dif v\Bigr)
\fu[u^{(1)}]_{a_1 a_2} \overline{\fu[u^{(1)}]}_{b_1 b_2}
\dif \lambda\\
+
\freqint \Bigl(\int_{-\infty}^{\infty} 
\psi^{(1)}(v, \lambda) 
\overline{\psi^{(2)}(v, -\lambda)} \dif v\Bigr)
\fu[u^{(1)}]_{a_1 b_2} \overline{\fu[u^{(1)}]}_{b_1 a_2}
\dif \lambda\\
+
\freqint \freqint 
\Bigl(\int_{-\infty}^{\infty} 
\psi^{(1)}(v, \lambda_1) 
\overline{\psi^{(2)}(v, -\lambda_2)} \dif v\Bigr)
\tilde{\gamma}_{a_1 a_2 b_1 b_2} (u^{(1)};\lambda_1, \lambda_2, -\lambda_2)
\dif \lambda_1 \dif \lambda_2
\Biggr),
\end{multline}
where $\tilde{\gamma}$ is the fourth-order spectral density of the process.
\end{lem}

\begin{rem}
The sequence $\beta_T$ is used to alleviate the divergence of the harmonic series.
There exists a sequence $\beta_T$ such that $\beta_T^{-1} = O(\log T)$
(See Remark \ref{rem:b7} below for details).
For this sequence, the condition $T^{-q/2} \beta_T^{-1} \to 0$ always holds true for $q \geq 2$.
\end{rem}

\begin{proof}
Using the expression \eqref{eq:rb9}, we have $\mathscr{A}_T(u)_{ab}$ as
\begin{equation} \label{eq:b11}
\mathscr{A}_T(u)_{ab}
=
\frac{1}{2\pi T}
 \sum_{k = 1}^{T}
\sum_{s = 1 - T}^{T - 1}
\hat{\psi}_T(u - u_k, -s)
\bx{[k + 1/2 + s/2]}^{(a)} \bx{[k + 1/2 - s/2]}^{(b)},
\end{equation}
which is a linear combination of $\bx{[k + 1/2 + s/2]}^{(a)} \bx{[k + 1/2 - s/2]}^{(b)}$.
We apply Lemma \ref{lem:a4} to compute the higher order cumulants.
Actually, it holds that
\begin{multline*}
\cum\bigl(
\xi_T(u^{(1)}; \psi^{(1)})_{a_1b_1},
\cdots,
\xi_T(u^{(q)}; \psi^{(q)})_{a_qb_q}
\bigr)\\
=
\cum\Bigl(
\frac{1}{T}
\sum_{\kappa_1}
\bx{\kappa_1}^{(a_1)} \bx{\kappa_1 - s_1}^{(b_1)},
\frac{1}{T}
\sum_{\kappa_2}
\bx{\kappa_2}^{(a_2)} \bx{\kappa_2 - s_2}^{(b_2)},
\dots,
\frac{1}{T}
\sum_{\kappa_q}
\bx{\kappa_q}^{(a_q)} \bx{\kappa_q - s_q}^{(b_q)}
\Bigr)\\
=
\frac{1}{T^q} \sum_{\kappa_1, \dots, \kappa_q}
\cum\Bigl(
\bx{\kappa_1}^{(a_1)} \bx{\kappa_1 - s_1}^{(b_1)},
\dots,
\bx{\kappa_q}^{(a_q)} \bx{\kappa_q - s_q}^{(b_q)}
\Bigr),
\end{multline*}
where for brevity, we let 
\[
\kappa_1 := [k_1 + 1/2 + s_1/2],
\quad
\kappa_2 := [k_2 + 1/2 + s_2/2],
\quad
\dots,
\quad
\kappa_q := [k_q + 1/2 + s_q/2].
\]
To compute higher order cumulants, 
we have to consider all indecomposable partitions of the following table 
(See \cite{brillinger1981}, Theorem 2.3.2):
\[
\begin{array}{cc}
\bx{\kappa_1}^{(a_1)} &  \bx{\kappa_1 - s_1}^{(b_1)} \\
\bx{\kappa_2}^{(a_2)} &  \bx{\kappa_2 - s_2}^{(b_2)} \\
\vdots & \vdots \\
\bx{\kappa_q}^{(a_q)} &  \bx{\kappa_q - s_q}^{(b_q)} \\
\end{array}.
\]
In view of Lemma \ref{lem:a4} with some tedious computation, 
all indecomposable partitions can be approximated by those cumulants of the stationary process
with a bias of lower order for a fixed $q \geq 2$.

We give a representative example of a partition below.
The other partitions can be evaluated in the same manner.
Without loss of generality, let $q$ be odd.
Suppose we evaluate the following cumulant:
\begin{equation*}
\frac{1}{T^q} \sum_{\kappa_1, \dots, \kappa_q}
\cum\Bigl( \bx{\kappa_1}^{(a_1)}, \bx{\kappa_2}^{(a_2)}\Bigr)
\cum\Bigl( \bx{\kappa_2 - s_2}^{(b_2)}, \bx{\kappa_3 - s_3}^{(b_3)}\Bigr)
\cdots
\cum\Bigl( \bx{\kappa_q}^{(a_q)}, \bx{\kappa_1 - s_1}^{(b_1)}\Bigr).
\end{equation*}
If we replace variables $\kappa_2, \dots, \kappa_q$ with
$\tau_2: = \kappa_2 - \kappa_1, \dots, \tau_q := \kappa_q - \kappa_1$, then we have
\begin{multline} \label{eq:b12}
\frac{1}{T^q} \sum_{\kappa_1, \tau_2, \dots, \tau_q}
\cum\Bigl( \bx{\kappa_1}^{(a_1)}, \bx{\kappa_1 + \tau_2}^{(a_2)}\Bigr)\\
\cum\Bigl( \bx{\kappa_1 + \tau_2 - s_2}^{(b_2)}, \bx{\kappa_1 + \tau_3 - s_3}^{(b_3)}\Bigr)
\cdots
\cum\Bigl( \bx{\kappa_1 + \tau_q}^{(a_q)}, \bx{\kappa_1 - s_1}^{(b_1)}\Bigr).
\end{multline}
Applying Lemma \ref{lem:a4}, \eqref{eq:b12} can be approximated by
\begin{equation} \label{eq:b13}
\frac{1}{T^q} \sum_{\kappa_1 = 1}^T
\sum_{\tau_2, \cdots, \tau_q}
\bm{\gamma}(u_{\kappa_1}, \tau_2)_{a_1 a_2}
\bm{\gamma}(u_{\kappa_1}, \tau_3 - s_3 - \tau_2 + s_2)_{b_2 b_3}
\cdots
\bm{\gamma}(u_{\kappa_1}, -s_1 - \tau_q)_{a_q b_1}.
\end{equation}
More precisely, the absolute bias between \eqref{eq:b12} and \eqref{eq:b13} is bounded by
\[
T^{-q} \sum_{i = 1}^q C_i\Bigl(
1 + \frac{1}{l(s_i)}
\Bigr).
\]
Returning back to the expression \eqref{eq:b11},
we see that the full expression of the absolute bias is bounded by
\begin{align*}
&
\frac{1}{2\pi b_T^q}
\sum_{s_1, \cdots, s_q = 1 - T}^{T - 1}
\prod_{i = 1}^q
\sup_u \abs{\hat{\psi}^{(i)}(u, -s_i)}
\Babs{
\frac{1}{T^q} \sum_{\kappa_1, \tau_2, \dots, \tau_q}
\cum\Bigl( \bx{\kappa_1}^{(a_1)}, \bx{\kappa_1 + \tau_2}^{(a_2)}\Bigr)\\
&
\cum\Bigl( \bx{\kappa_1 + \tau_2 - s_2}^{(b_2)}, \bx{\kappa_1 + \tau_3 - s_3}^{(b_3)}\Bigr)
\cdots
\cum\Bigl( \bx{\kappa_1 + \tau_q}^{(a_q)}, \bx{\kappa_1 - s_1}^{(b_1)}\Bigr)\\
&
-\frac{1}{T^q} \sum_{\kappa_1 = 1}^T
\sum_{\tau_2, \cdots, \tau_q}
\bm{\gamma}(u_{\kappa_1}, \tau_2)_{a_1 a_2}
\bm{\gamma}(u_{\kappa_1}, \tau_3 - s_3 - \tau_2 + s_2)_{b_2 b_3}
\cdots
\bm{\gamma}(u_{\kappa_1}, - s_1 - \tau_q)_{a_q b_1}
}\\
=&
O(T^{-q} b_T^{-q} \beta_T^{-q}).
\end{align*}
In summary, all cumulants of order $q$ for $\mathscr{A}_T$ can be approximated by those of 
the stationary process with a bias of order $O(T^{-q} b_T^{-q} \beta_T^{-q})$.
Thus, the bias in those cumulants for $\xi_T$ is $O(T^{-q/2} b_T^{-q/2} \beta_T^{-q})$.
Furthermore, it holds that
\begin{equation} \label{eq:b14}
\cum\bigl(
\xi_T(u^{(1)}; \psi^{(1)})_{a_1b_1},
\cdots,
\xi_T(u^{(q)}; \psi^{(q)})_{a_qb_q}
\bigr)
=
O(T^{1-q/2} b_T^{1-q/2}),
\end{equation}
since $\psi^{(1)}(\cdot, \lambda)$, \dots, $\psi^{(q)}(\cdot, \lambda)$ 
are all functions of bounded variation.
Therefore, the bias is asymptotically negligible.
A representative example of \eqref{eq:b14} is shown below.

Let us consider the case $q = 2$ for $\xi_T$. Note that $q$ is even now. 
We have three terms of the type \eqref{eq:b13}
, i.e.,
\begin{enumerate}[(i)]
\item the approximation for $\cum(\bx{\kappa_1}^{(a_1)},  \bx{\kappa_1 - s_1}^{(b_1)})
\cum(\bx{\kappa_2}^{(a_2)}, \bx{\kappa_2 - s_2}^{(b_2)})$:
\begin{multline} \label{eq:b15}
\frac{b_T}{T} \sum_{\kappa_1= 1}^T \sum_{s_1, s_2, \tau_2} 
\hat{\psi}_T^{(1)}(u^{(1)} - u_{\kappa_1}, s_1)
\hat{\psi}_T^{(2)}(u^{(2)} - u_{\kappa_1}, s_2) \\
\bm{\gamma}(u_{\kappa_1}, \tau_2)_{a_1 a_2}
\bm{\gamma}(u_{\kappa_1}, \tau_2 - s_2 + s_1)_{b_1 b_2};
\end{multline}
\item the approximation for $\cum(\bx{\kappa_1}^{(a_1)},  \bx{\kappa_2 - s_2}^{(b_2)})
\cum(\bx{\kappa_1 - s_1}^{(b_1)}, \bx{\kappa_2}^{(a_2)})$:
\begin{multline} \label{eq:b16}
\frac{b_T}{T} \sum_{\kappa_1= 1}^T \sum_{s_1, s_2, \tau_2} 
\hat{\psi}_T^{(1)}(u^{(1)} - u_{\kappa_1}, s_1)
\hat{\psi}_T^{(2)}(u^{(2)} - u_{\kappa_1}, s_2) \\
\bm{\gamma}(u_{\kappa_1}, \tau_2 - s_2)_{a_1 b_2}
\bm{\gamma}(u_{\kappa_1}, \tau_2 + s_1)_{b_1 a_2};
\end{multline}
\item the approximation for $\cum(\bx{\kappa_1}^{(a_1)},  \bx{\kappa_2 - s_2}^{(b_2)},
\bx{\kappa_1 - s_1}^{(a_2)}, \bx{\kappa_2}^{(b_1)})$:
\begin{multline} \label{eq:b17}
\frac{b_T}{T} \sum_{\kappa_1= 1}^T \sum_{s_1, s_2, \tau_2} 
\hat{\psi}_T^{(1)}(u^{(1)} - u_{\kappa_1}, s_1)
\hat{\psi}_T^{(2)}(u^{(2)} - u_{\kappa_1}, s_2) \\
\gamma_{a_1 a_2 b_1 b_2}(u_{\kappa_1}; -s_1, \tau_2, \tau_2 - s_2).
\end{multline}
\end{enumerate}

We first explain the term \eqref{eq:b15}.
By repeated application of  the Parseval equality 
(see, e.g., the proof of Lemma 2.2 in \cite{ht1982} for details)
and by Lemma P5.1 in \cite{brillinger1981},
the term \eqref{eq:b15} is equivalent to
\begin{multline*}
2\pi b_T
\int_{0}^1 \freqint 
\psi^{(1)}_T(u^{(1)} - u, \lambda) 
\overline{\psi^{(2)}_T(u^{(2)} - u, \lambda)}\\ 
\fu_{a_1 a_2} \overline{\fu}_{b_1 b_2}
\dif \lambda \dif u
+
O(T^{-1}b_T^{-1}).
\end{multline*}
Under Assumption \ref{asp:b2}, if $u^{(1)} \not = u^{(2)}$, we have
\[
\int_{0}^1 \freqint 
\psi^{(1)}_T(u^{(1)} - u, \lambda) 
\overline{\psi^{(2)}_T(u^{(2)} - u, \lambda)}\\ 
\fu_{a_1 b_2} \overline{\fu}_{b_1 a_2}
\dif \lambda \dif u
=
o(b_T^{-1}),
\]
since the supports of $\psi^{(1)}$ and $\psi^{(2)}$ are compact.
Thus, the term \eqref{eq:b15} converges to
\begin{equation} \label{eq:b18}
2\pi \delta(u^{(1)}, u^{(2)})
\freqint \Bigl(\int_{-\infty}^{\infty} 
\psi^{(1)}(v, \lambda) 
\overline{\psi^{(2)}(v, \lambda)} \dif v\Bigr)
\fu[u^{(1)}]_{a_1 a_2} \overline{\fu[u^{(1)}]}_{b_1 b_2}
\dif \lambda,
\end{equation}
where $\delta$ is a delta function such that $\delta (a, b) = 1$ if $a = b$, and 0 otherwise.
Similarly, the term \eqref{eq:b16} converges to
\begin{equation} \label{eq:b19}
2\pi \delta(u^{(1)}, u^{(2)})
\freqint \Bigl(\int_{-\infty}^{\infty} 
\psi^{(1)}(v, \lambda) 
\overline{\psi^{(2)}(v, -\lambda)} \dif v\Bigr)
\fu[u^{(1)}]_{a_1 b_2} \overline{\fu[u^{(1)}]}_{b_1 a_2}
\dif \lambda.
\end{equation}
The term \eqref{eq:b17} converges to 
\begin{multline} \label{eq:b20}
2\pi \delta(u^{(1)}, u^{(2)})
\freqint \freqint 
\Bigl(\int_{-\infty}^{\infty} 
\psi^{(1)}(v, \lambda_1) 
\overline{\psi^{(2)}(v, -\lambda_2)} \dif v\Bigr)
\\
\tilde{\gamma}_{a_1 a_2 b_1 b_2} (u^{(1)};\lambda_1, \lambda_2, -\lambda_2)
\dif \lambda_1 \dif \lambda_2,
\end{multline}
by repeated application of the Parseval equality. 
Combining all terms \eqref{eq:b18}, \eqref{eq:b19} and \eqref{eq:b20}, 
we obtain the results of Lemma \ref{lem:b5}.
\end{proof}

\subsection{Asymptotic Normality} \label{subsec:b2}
Here, we show the asymptotic normality of the empirical spectral process $\xi_T(\psi)$ in \eqref{eq:b5}.
To this goal, we adopt the idea in \cite{dp2009}
to use the Gaussian kernel as the mollifier with the property of being rapidly decreasing.
Let $G$ be the Gaussian kernel, that is,
\[
G(x): = \frac{1}{\sqrt{2\pi}} \exp\Bigl(
-\frac{1}{2} x^2
\Bigr),
\]
and $G_b$ the mollifier 
\[
G_{\beta}(x) = \frac{1}{\beta} G\Bigl(
\frac{x}{\beta}
\Bigr),
\]
with $\beta: = \beta_T \to 0$ as $T \to \infty$.
From the convolution theorem, 
the Fourier coefficients $\hat{\psi}^{*T}$ of $\psi^{*T} := \psi * G_{\beta}$ are
\begin{equation} \label{eq:b21}
\hat{\psi}^{*T}(u, k) = 
\hat{\psi}(u, k) \hat{G}_{\beta}(k), \quad k \in \Z.
\end{equation}

\begin{rem} \label{rem:b7}
The remarkable feature of this manipulation is that
\begin{equation*}
\sum_{k \in \Z}
\sup_{u \in [0, 1]} \abs{\hat{\psi}^{*T}(u, k)}
\leq
\sum_{k \in \Z}
\sup_{u \in [0, 1]} \abs{\hat{\psi}(u, k)},
\end{equation*}
since for any fixed $k \in \Z$, 
\[
\abs{\hat{G}_{\beta}(k)} = \Babs{\exp\Bigl(
\frac{-\beta^2 k^2}{2}
\Bigr)}
\leq 
1.
\]
In addition, the following result holds.
\begin{equation} \label{eq:b22}
\sum_{k \in \Z}
\sup_{u \in [0, 1]} \abs{\hat{\psi}^{*T}(u, k)}
=
O\Bigl(\log\bigl( \beta_T^{-1} \bigr)\Bigr).
\end{equation}
If we take $\beta_T$ as $\beta_T = T^{-k}$ for any $k \geq 1$, then we have
\[
\sum_{k \in \Z}
\sup_{u \in [0, 1]} \abs{\hat{\psi}^{*T}(u, k)}
=
O(\log T).
\]  
\end{rem}

\begin{proof}[Proof of Remark \ref{rem:b7}]
For any $1 \leq i, j \leq p$, let $\psi: = \phi_{ij} \in \Psi$ as in Assumption \ref{asp:b1}.
Note that $\psi(u, \cdot)$ is a continuous function of bounded variation.
\begin{enumerate}[(i)]
\item Let $k \not = 0$.
From Jordan decomposition theorem, there exists a signed measure $g_{\psi}$ such that
\[
\hat{\psi}(u, k)
=
\freqint 
\frac{\exp(-\I k \lambda) - 1}{-\I k} 
g_{\psi}(u, \dif \lambda),
\]
which leads to
\begin{equation} \label{eq:b23}
\sup_{u \in [0, 1]} \abs{\hat{\psi}(u, k)} \leq
\frac{C}{\abs{k}} \sup_{u \in [0, 1]} V\bigl(
\psi(u, \cdot)
\bigr)
\leq
\frac{C \tau_{\infty, \rm{TV}}}{\abs{k}}.
\end{equation} 
\item Let $k = 0$. 
\begin{equation} \label{eq:b24}
\sup_{u \in [0, 1]} \abs{\hat{\psi}(u, 0)}
\leq
2\pi  \sup_{u \in [0, 1]}\sup_{\lambda \in [-\pi, \pi]} \psi(u, \lambda)
\leq
2 \pi \tau_{\infty, \infty}.
\end{equation}
\end{enumerate}
Combing \eqref{eq:b23} and \eqref{eq:b24} with the relation \eqref{eq:b21}, we obtain
\[
\sup_{u \in [0, 1]} \abs{\hat{\psi}^{*T}(u, k)} \leq C\Bigl(
1 + 
\sum_{k = 1}^{\infty} \frac{1}{\abs{k}} \exp\Bigl(
\frac{-\beta^2 k^2}{2}
\Bigr)
\Bigr)
=
O\Bigl(\log\bigl( \beta^{-1} \bigr)\Bigr).
\]
Thus, the equation \eqref{eq:b22} is shown.
\end{proof}

Next result shows that the $\xi_T(u_k)_{ab}$ converges in finite dimensional distributions for $k \geq 1$.
\begin{thm} \label{thm:b8}
Suppose Assumptions \ref{asp:a1} and \ref{asp:b1} hold.
Let $b_T \to 0$ and $T b_T \to \infty$, as $T \to \infty$.
For any $q$, and $u^{(1)}, \dots, u^{(q)} \in [0, 1]$, it holds that
\[
\big(\xi_T(u^{(1)}; \psi^{(1)})_{a_1b_1}, \cdots, \xi_T(u^{(q)}; \psi^{(q)})_{a_qb_q} \bigr)^{\top}
\dlim
\mathcal{N}\bigl(\bzero, (V_{jk})_{j, k = 1, \dots q}\bigr),
\qquad
\text{as $T \to \infty$,}
\]
where $V_{jk}$ is
\begin{multline*}
V_{jk} =
2\pi \delta(u^{(j)}, u^{(k)})\\ 
\Biggl(
\freqint \Bigl(\int_{-\infty}^{\infty} 
\psi^{(j)}(v, \lambda) 
\overline{\psi^{(k)}(v, \lambda)} \dif v\Bigr)
\fu[u^{(j)}]_{a_j a_k} \overline{\fu[u^{(j)}]}_{b_j b_k}
\dif \lambda\\
+
\freqint \Bigl(\int_{-\infty}^{\infty} 
\psi^{(j)}(v, \lambda) 
\overline{\psi^{(k)}(v, -\lambda)} \dif v\Bigr)
\fu[u^{(j)}]_{a_j b_k} \overline{\fu[u^{(j)}]}_{b_j a_k}
\dif \lambda\\
+
\freqint \freqint 
\Bigl(\int_{-\infty}^{\infty} 
\psi^{(j)}(v, \lambda_1) 
\overline{\psi^{(k)}(v, -\lambda_2)} \dif v\Bigr)
\tilde{\gamma}_{a_j a_k b_j b_k} (u^{(j)};\lambda_1, \lambda_2, -\lambda_2)
\dif \lambda_1 \dif \lambda_2
\Biggr),
\end{multline*}
where $\tilde{\gamma}$ is the fourth-order spectral density of the process.
\end{thm}

\begin{proof}
First we show that
\[
\var\Bigl(
\xi_T(u ; \psi)_{ab} - \xi_T(u ; \psi^{*T})_{ab}
\Bigr)
\to
0,
\]
which, in turn, shows that 
\begin{equation} \label{eq:b25}
\xi_T(u ; \psi)_{ab}  - \xi_T(u ; \psi^{*T})_{ab} \plim 0.
\end{equation} 
As in Remark \ref{rem:b7}, let $\beta_T = T^{-k}$ for any $k \geq 1$.
Following this choice, we have $O(\beta_T/b_T) = o(1)$.

Note that
\begin{align*}
&\var\Bigl(
\xi_T(u ; \psi)_{ab} - \xi_T(u ; \psi^{*T})_{ab}
\Bigr) \\
&=
Tb_T 
\var\Bigl(
\frac{1}{2\pi T}\\
& \qquad
\sum_{k = 1}^{T}
\sum_{s = 1 - T}^{T - 1}
\bigl\{
\hat{\psi}_T(u - u_k, -s)
-
\hat{\psi}^{*T}_T(u - u_k, -s)
\bigr\}
\sum_{t \in \mathcal{T}_s}
\bx{k + t}^{(a)} \bx{k+ t + s}^{(b)}
\Bigr) \\
&\leq
b_T^{-1} 
\Bigl(\sup_{u} \tsum[s]{\infty} \abs{\hat{\psi}(u, -s)
-
\hat{\psi}^{*T}(u, -s)}
\Bigr)^2
\\
&\leq
C b_T^{-1} \tsum[s]{\infty} \frac{\abs{\exp(- s^2 \beta_T^2/2) - 1}^2}{s^2},
\end{align*}
where the last inequality follows from \eqref{eq:b21}.
Since $\abs{\exp(- s^2 \beta_T^2/2) - 1} \leq \min (1,  s^2 \beta_T^2/2)$,
the order of the last term is $O(\beta_T/b_T) = o(1)$.
Thus, \eqref{eq:b25} is shown.

Now, we only have to consider the finite distributions of $\xi_T(u ; \psi^{*T})$.
However, from Remark \ref{rem:b7}, we find that the condition \eqref{eq:b9} is satisfied
and thus the covariance matrix of $\xi_T(u ; \psi^{*T})$ can be expressed in the form of 
\eqref{eq:b10}.
Therefore, the proof is completed.
\end{proof}

Finally, remembering the matrix $\bp$ satisfies Assumption \ref{asp:b1},
we define $\mathscr{A}_T^{\circ}(u)$ and $\bar{\mathscr{A}}_T^{\circ}(u)$ as
\begin{align*}
\mathscr{A}_T^{\circ}(u)
&:= 
\frac{1}{T} \sumn[k]  \freqint \bp_T(u - u_k, \lambda) \bI_T(u_k, \lambda) \dif\lambda, \\
\bar{\mathscr{A}}_T^{\circ}(u)&:=
\frac{1}{T}  \sumn[k] \freqint \bp_T(u - u_k, \lambda) \fu[u_k] \dif \lambda,
\end{align*}
and let $\zeta_T(u)$ be 
\begin{equation} \label{eq:b30}
\zeta_T(u) = \sqrt{Tb_T} \,\, \tr
\bigl(\mathscr{A}_T^{\circ}(u) - \bar{\mathscr{A}}_T^{\circ}(u)\bigr).
\end{equation}

\begin{cor} \label{cor:b11}
Suppose Assumptions \ref{asp:a1}, \ref{asp:b1} and \ref{asp:b2} hold.
If $b_T = o(1)$ and $b_T^{-1} = o\bigl(
T(\log T)^{-6}
\bigr)$,
then it holds that
\[
\big(\zeta_T(u^{(1)}), \cdots, \zeta_T(u^{(q)}) \bigr)^{\top}
\dlim
\mathcal{N}\bigl(\bzero, (\tilde{V}_{jk})_{j, k = 1, \dots q}\bigr),
\qquad
\text{as $T \to \infty$,}
\]
where $\tilde{V}_{jk}$ is given by
\begin{multline} \label{eq:b31}
\tilde{V}_{jk}
=
4\pi \delta(u^{(j)}, u^{(k)})
\Biggl(
\freqint 
\tr \Bigl(\int_{-\infty}^{\infty} 
\fu[u^{(j)}] \bp(v, \lambda) \fu[u^{(j)}] \bp(v, \lambda)
\dif v\Bigr)\dif \lambda\\
+
\frac{1}{2}
\sum_{r, t, u ,v = 1}^p
\freqint \freqint 
\Bigl(\int_{-\infty}^{\infty} 
\bp_{rt}(v, \lambda_1) 
\bp_{uv}(v, \lambda_2)\\
\tilde{\gamma}_{rtuv} (u^{(j)}; -\lambda_1, \lambda_2, -\lambda_2)
\dif v\Bigr) \dif \lambda_1 \dif \lambda_2
\Biggr),
\end{multline}
where $\tilde{\gamma}$ is the fourth-order spectral density of the process.
\end{cor}

\begin{proof}
From the definition of $\zeta_T(u)$ in \eqref{eq:b30},
we see that $\zeta(u)$ is a linear combination of the processes $\xi_T(u)$ in \eqref{eq:b5}.
With a similar computation to the latter part in Lemma A.3.3. in \cite{ht1982},
we obtain \eqref{eq:b31}.
\end{proof}

\end{document}